\documentclass[final,3p,times,twocolumn]{elsarticle}

\usepackage{amssymb}
\usepackage{booktabs} % For formal tables
\usepackage{epstopdf}
\usepackage{balance}
\usepackage{xcolor}

\newtheorem{defi}{Definition}[section]
\newtheorem{theorem}{Theorem}[section]

\newtheorem{prop}[theorem]{Proposition}

\newtheorem{proof}[theorem]{Proof}

\begin{document}

\begin{frontmatter}



\title{Incremental Clustering Techniques for \\ Multi-Party Privacy-Preserving Record Linkage}

\author[label1,label2]{Dinusha Vatsalan\corref{cor1}}
\address[label1]{Data61, CSIRO, Eveleigh, NSW 2015, Australia}
\address[label2]{Research School of Computer Science, The Australian National University, Canberra, ACT 2600, Australia}
\address[label3]{Institut f\"ur %{$\ddot{\mbox{u}}$}r
  Informatik, Universit\"at % $\ddot{\mbox{a}}$}t
  Leipzig, Leipzig 04109, Germany}
\cortext[cor1]{Corresponding author.}
\ead{dinusha.vatsalan@data61.csiro.au}

\author[label2]{Peter Christen}
\ead{peter.christen@anu.edu.au}

\author[label3]{Erhard Rahm}
\ead{rahm@informatik.uni-leipzig.de}

\begin{abstract}
Privacy-Preserving Record Linkage (PPRL) supports the integration of sensitive information from multiple datasets, in particular the privacy-preserving matching of records referring to the same entity. PPRL has gained much attention in many application areas, with the most prominent ones in the healthcare domain. PPRL techniques tackle this problem by conducting linkage on masked (encoded) values. Employing PPRL on records from multiple (more than two) parties/sources (multi-party PPRL, MP-PPRL) is an increasingly important but challenging problem that so far has not been sufficiently solved. Existing MP-PPRL approaches are limited to finding only those entities that are present in all parties thereby missing entities that match only in a subset of parties. Furthermore, previous MP-PPRL approaches face substantial scalability limitations due to the need of a large number of comparisons between masked records. We thus propose and evaluate new MP-PPRL approaches that find matches in any subset of parties and still scale to many parties. Our approaches maintain all matches within clusters, where these clusters are incrementally extended or refined by considering records from one party after the other. An empirical evaluation using multiple real datasets ranging from $3$ to $26$ parties each containing up to $5$ million records validates that our protocols are efficient, 
%require around $5,000$ seconds to link ten databases each containing $1M$ records, 
and significantly outperform existing MP-PPRL approaches in terms of linkage quality and scalability.

\end{abstract}


\end{frontmatter}



\section{Introduction}
\label{sec:Introduction}

%Integrating data from multiple sources with the aim to identify matching pairs (from two sources) or matching sets (from more than two sources) of records that correspond to the same real-world entity is a crucial data pre-processing task for quality data mining and analytics~\cite{Chr12}. Various real-world applications require record linkage to improve data quality and enable accurate decision making. Example applications come from healthcare, businesses, the social sciences, and national security~\cite{Vat13}. 
%For example, health outbreak systems require data to be integrated from several sources, including human health data, travel data, consumed drug data, and even animal health data, to allow the early detection of infectious diseases before they spread widely around a country or even worldwide. %~\cite{Clif04}.
%Another contemporary example is
%national security applications that integrate data
%from law enforcement agencies, Internet service
%providers, businesses, as well as financial 
%institutions 
%to enable the accurate
%identification of crime and fraud, or of terrorism suspects. %~\cite{Phu12}.

%In the current era of Big Data, 
With the widespread collection of large-scale person-specific databases by many organizations, %many applications and research studies increasingly aim to leverage such data by integrating and linking multiple large databases (held by different parties) in order to identify matching records that correspond to the same real-world entity (patient)~\cite{Vat17b,Clif02,Che19}. 
multiple large databases (held by different parties) often need to be integrated and linked to identify matching records that correspond to the same real-world entity~\cite{Vat17b,Clif02,Che19} %, %while protecting the sensitive or confidential information of the entities  
for viable data mining and analytics applications. 
%Example applications range from healthcare, businesses, the social sciences, and national security~\cite{Vat13}. 
%For example, health outbreak systems require data to be integrated from several sources, including human health data, travel data, consumed drug data, and even animal health data, to allow the early detection of infectious diseases before they spread widely around a country or even worldwide~\cite{Clif04}.
%Another contemporary example is
%national security applications that integrate data
%from law enforcement agencies, Internet service
%providers, businesses, as well as financial 
%institutions 
%to enable the accurate
%identification of crime and fraud, or of terrorism suspects~\cite{Phu12}.
%
The absence of unique entity identifiers across different databases requires using commonly available personal identifying attributes, such as names and addresses, for integrating and linking records from those databases. The values in these quasi-identifiers (QIDs) are often dirty, i.e.\ contain errors and variations, or they can be missing, which makes the linkage task challenging~\cite{Chr12,Chi17}.
In addition, such attributes often contain sensitive personal information about the entities to be linked, and therefore sharing or exchanging such values among different organizations is often prohibited due to privacy and confidentiality concerns~\cite{Dur13,Kara15,Vat13}. 
Addressing these challenges, privacy-preserving record linkage (PPRL) has attracted increasing interest over the last two decades~\cite{Vat17b,Vat13} and been employed in several real applications.

%Known as multi-party PPRL (MP-PPRL), this research direction has received increasing attention in recent times~\cite{Ran14,Vat14c,Vat16b}. 
%PPRL has has received increasing attention in recent times and been employed in several real applications. %mainly in the health domain.
For example, data from hospitals and clinical registries were linked with data from central cancer registries and from the Australian Bureau of Statistics using PPRL techniques for a study on surgical treatment received by aboriginal and non-aboriginal people with lung cancer~\cite{Con04}. Data from several cantonal and national registries were linked in Switzerland using Bloom filter-based PPRL to investigate long-term consequences of childhood cancer~\cite{Kue11}. %Pseudo-anonymisation-based PPRL techniques are being used in Germany for several healthcare research projects~\cite{}.
In 2016, the Interdisciplinary Committee of the International Rare Diseases Research Consortium launched a task team to explore approaches to PPRL for linking several genomic and clinical data sets~\cite{Bak18}.

%\newpage
Further, the Office for National Statistics (ONS) in the UK established the program `Beyond 2011' to carry out research to study the options for production of population and socio-demographics statistics for England and Wales, by linking anonymous data to ensure that high levels of privacy of data about people are maintained~\cite{ONS13}.
Another application of PPRL in the domain of national security is to integrate data from law enforcement agencies, Internet service providers, businesses, as well as financial institutions, to enable %the %accurate 
identifying crime and fraud, or of terrorism suspects~\cite{Phu12}.

%Peter added 2014020
The majority of linkage techniques and frameworks have been developed for linking records from only two databases~\cite{Vat13,Kop10,Wan07}.
%Some recent work has began to investigate blocking techniques
%for multi-party PPRL~\cite{Ran15,Ran16}.
It is not trivial to extend existing PPRL techniques to multiple databases by sending the encoded
databases from all parties to a Linkage Unit ($LU$), where a $LU$ is an external party that has been used in several existing PPRL approaches for conducting or facilitating the linkage of encoded records sent to it by the database owners~\cite{Vat13}. 
%One basic approach would be to send the encoded
%databases from all parties to a Linkage Unit ($LU$), which is an external party that has been used in several existing PPRL approaches for conducting or facilitating the linkage of encoded records sent to it by the database owners~\cite{Vat13}. 
%One approach would be for the $LU$ to link each pair of encoded databases, but the number of such pairs grows quadratically with the number of databases. 
%A second approach would be to combine all encoded databases into one large database and for the $LU$ to perform a deduplication~\cite{Chr12} of records in this large database. 
At the $LU$, it would then become necessary to determine 
pair-wise similarities between records and to group similar records into clusters where one cluster is assumed to represent one entity~\cite{Has09c}. 
Only few basic grouping/clustering techniques have been described for multi-database linkage, with each of them having limitations as discussed in detail in Section~\ref{sec-related}. 
Such clustering schemes have been studied for general record linkage \cite{Has09c,Nan19,Sae18} but have received almost no attention so far for PPRL.  %(see Section~\ref{sec-related}). 
Furthermore, sending all the encoded records from multiple
parties to the $LU$ has privacy risks. For example, with
Bloom filter-based encoding~\cite{Sch09} (to be described
in the next section), the more
Bloom filters the $LU$ receives the more likely it will be
able to attack these Bloom filter databases using
cryptanalysis attacks because more frequency information will
become available that can be exploited~\cite{Chr18,Chr18b}.
Only few techniques have been developed that can perform multi-party linkage in a privacy-preserving context (i.e.\ MP-PPRL). The main drawbacks of these small number of existing MP-PPRL approaches are that they either (1) only consider the blocking step to reduce the matching space~\cite{Chr11} but not how the matching is done,
%(as will be described in detail in Section~\ref{subsec:process}), 
(2) only support exact matching, which classifies record sets as matches if their masked QIDs are exactly the same~\cite{Chr12}, (3) are applicable to QIDs of categorical data only (however, linkage using QIDs of string data, such as names and addresses, is required in many real applications~\cite{Vat13,Vat14c}), or (4) they do not support subset matching where records that match across subsets of databases also need to be identified in addition to records that match across all databases. %Neither of these techniques allow subset matching for MP-PPRL where records that match across subsets of databases are also identified in addition to records that match across all databases.
The primary challenge of MP-PPRL is the complexity of linkage, 
which generally is exponential with the number of databases
to be linked and their sizes~\cite{Vat16b}. This challenge
multiplies when matching records from any possible subset
across databases need to be identified.
%
%
%To the best of our knowledge, no work has been done so far for efficient subset matching for MP-PPRL. %where records that match across subsets of databases also need to be identified in addition to records that match across all databases.
%
%Furthermore, 
%Most existing private comparison and classification techniques for MP-PPRL either support exact matching only, or they are applicable only to categorical data~\cite{Kan08,Kara15,Lai06,Kee04}. In general, record linkage requires comparing string values in quasi-identifiers (QIDs), which are partially identifying attributes, such as names and addresses. These QIDs often contain data errors and variations, and therefore exact matching of categorical data is not sufficient for effective linkage~\cite{Vat13}. Recent work on MP-PPRL has addressed this problem by developing approximate string comparison functions for MP-PPRL, for example using Bloom filter (BF)~\cite{Vat14c} and counting Bloom filter (CBF)~\cite{Vat16b}. %However, these approaches are not scalable to multiple large databases due to their exponential number of comparisons required. 

\smallskip
\textbf{Contributions:}~ 
In this paper, we propose an efficient and scalable MP-PPRL protocol that allows subset matching between multiple large databases using a $LU$. 
%Our approach is based on a linkage unit ($LU$). 
$LU$-based approaches for PPRL are well suited for efficient linking of multiple large databases for practical applications, as the number of communication steps required among the database owners, as well as the risk of information leakage from a sensitive database to other database owners, are reduced when a $LU$ is used~\cite{Vat13}. 

We develop two variations of incremental clustering combined with a graph-based linkage for MP-PPRL where clusters of encoded records %(masked using Bloom filter encoding or variations~\cite{Sch15}, as described in Section~\ref{sec-preli}) 
are iteratively merged and refined such that the output of clusters are the matching sets of records (i.e.\ each cluster represents a set of matching records that correspond to the same entity). 
Clustering-based approaches are deemed most suitable for holistic data integration, and have been used in several non-PPRL approaches for scaling data integration to many sources~\cite{Rah16,Ran15b}.
%DV - 11/02/2019
% PC 14/02/2019
Compared to greedy mapping~\cite{Ken98} (as described in Section~\ref{sec-protocol}), our proposed incremental clustering methods perform significantly better in terms of linkage quality. 

We use counting Bloom
filter-based encoding~\cite{Vat16b} which has lower risk of privacy leakage as the frequency information available in counting Bloom filters is significantly less than basic Bloom filters~\cite{Vat16b}. Additionally, the risk
of collusion between different parties and the $LU$ can be
reduced in our incremental clustering approach by using
different encoding parameters in different iterations,
as we discuss in Section~\ref{subsec_privacy_analsis}.

We provide a comprehensive evaluation of our proposed approach which shows that it has a quadratic computation complexity in the size and the number of the databases that are linked. This complexity
%of our approach
is significantly lower compared to the exponential complexity of existing MP-PPRL approaches~\cite{Vat14c,Vat16b,Lai06}, as we theoretically and empirically validate in Sections~\ref{sec-analysis} and~\ref{sec-experiment} using large real voter and health datasets.

\smallskip
\textbf{Outline:}~ 
In Section~\ref{sec-preli} we provide the required preliminaries and in Section~\ref{sec-protocol} we describe our protocol for MP-PPRL. We analyze our protocol in terms of complexity, privacy, and linkage quality in Section~\ref{sec-analysis}, and validate these analyses through an empirical evaluation in Section~\ref{sec-experiment}. 
We discuss related work in MP-PPRL in Section~\ref{sec-related}. 
Finally, we conclude the paper with an outlook to future research directions in Section~\ref{sec-conclusion}.

% --------------------------------------------------------------------

\section{Preliminaries}
\label{sec-preli}

In this section, we define 
the problem of MP-PPRL and describe
the preliminaries required for our protocol.

\begin{defi}[\textbf{MP-PPRL}] 
Assume $P_{1}, \ldots, P_{p}$ are the $p$ owners (parties) of the \emph{deduplicated} databases $\mathbf{D}_{1}, \ldots, \mathbf{D}_{p}$, respectively. MP-PPRL allows the parties $P_i$ to determine which of their records $r_{i,x} \in \mathbf{D}_{i}$ match with records in other database(s) $r_{j,y} \in \mathbf{D}_{j}$ with $1 \le i,j \le p$ and $j \neq i$ based on the (masked or encoded) quasi-identifiers (QIDs) of these records. The output of this process is a set $\mathbf{M}$ of match clusters, where a match cluster $c \in \mathbf{M}$ contains %a set of matching records of 
a maximum of one record from each database and $1 < |c| \le p$. Each $c \in \mathbf{M}$ is identified as a set of matching records representing the same real-world entity.
%The decision model $C(\cdot)$ classifies sets $\mathbf{S}=(S_1,S_2,\cdots,S_z)$ of records from different parties $S_k \in \mathbf{S}=(r_{i,x}$, $\cdots$, $r_{j,y})$ with $1 \le i,j \le p$, $i \neq j$ and $2 \le |S| \le p$, into one of the two classes $\mathbf{M}$ of matches and $\mathbf{U}$ of non-matches. A record set $S_k \in \mathbf{M}$ contains records from different parties that all are assumed to refer to the same entity, while $S_l \in \mathbf{U}$ contains records assumed to refer to different entities.
%A linkage unit ($LU$) is generally employed to conduct PPRL using the masked QID values of records sent by the database owners.
%Assuming the honest-but-curious adversary model, the parties $P_{1}, \ldots, P_{p}$ and the $LU$ (if used) are honest, in that they follow the linkage protocol steps~\cite{Lin09}. 
The parties do not wish to reveal their actual records with any other party. They however are prepared to disclose to each other, or to an external party (such as a researcher), the actual values of some selected attributes of the record sets that are in $\mathbf{M}$ to allow further analysis. 
\end{defi}

We assume that the individual databases do not contain any duplicates (i.e. multiple records about the same patient). Each party performs the necessary pre-processing steps including deduplication to ensure the quality of their own database. Many deduplication techniques have been developed in the literature~\cite{Chr12,Nau10} which can be used for deduplicating individual databases before linking them across different parties (such that there is only one record per entity/patient in a database, and therefore a record in one database can match to only one record in another database).
%A record in database $\mathbf{D}_i$ is denoted as $r_{i,x}$ 
%with $1 \le i \le p$ and $1 \le x \le |\mathbf{D}_i|$.
%
%\subsection{Data masking or encoding} 
%\label{subsec:masking}

We also assume that a private blocking, indexing, or filtering
technique is being used by the database
owners~\cite{Vat14c,Ran14,All05,Seh15}. Such techniques are being
used in general linkage and PPRL to reduce the number of comparisons
by grouping records according to a certain criteria and limiting
the comparison only to the records in the same
group~\cite{Chr12,Vat13}, or by pruning record pairs/sets that are
potential non-matches according to some criteria~\cite{Seh15}. Note
that blocking is not a focus of our paper, and that we assume that
the private blocking technique used by the database owners is
secure~\cite{Vat17b}.

%%Several private blocking techniques have recently been developed for MP-PPRL addressing the scalability challenge~\cite{Ran14,Ran16}. 
%However, even when a private blocking or filtering technique is used, the number of comparisons required between records remains exponential with an increasing number of databases in MP-PPRL, as h%as been empirically validated~\cite{Ran14,Vat14c}. Existing MP-PPRL techniques therefore have a high computation complexity, and so far they do not support subset matching (identifying records that are matching across any subset of databases, for example any three out of ten hospitals) rendering them impractical in real applications. 

Since QIDs that are generally used for linking (e.g. names and addresses) contain personal and sensitive information about individuals, PPRL needs to be conducted on the encoded or masked versions of these QIDs.
Any masking (encoding) function $mask(\cdot)$ can be used in our privacy-preserving linkage protocol to encode attribute values, as long as
the same $mask(\cdot)$ function is used by all database owners $P_i$
to mask their databases $\mathbf{D}_i$ into $\mathbf{D}_i^M$,
where $1 \le i \le p$.
%We assume $mask(\cdot)$ is unknown to $LU$.
%Two major categories of masking functions 
%developed for privacy-preserving algorithms are: 
%(1) cryptographic-based secure multi-party computation (SMC)
%techniques and (2) perturbation-based efficient masking techniques~\cite{Vat14}.
%Techniques developed under the former category are generally more expensive
%with regard to the computation and communication complexities
%while providing strong privacy guarantees and high linkage quality~\cite{Vat13,Yak09}. %{Lin00}. 
%%while
%Perturbation techniques, on the other hand, 
%efficiently perturb or modify the original
%values (to preserve privacy) 
%while still allowing to perform approximate matching between the
%masked values using the functional relationship between original
%and masked data~\cite{Vat13,Vat14}.
%
%Since Bloom filter (BF) encoding 
%%(as will be described in detail in %Section~\ref{subsec:masking}) 
%has widely been used for data masking in both research
%and practical applications of PPRL~\cite{Vat13,Ran13}, 
%we consider BF encoding for data masking
%in our approach.
We describe our protocol using the Bloom filter (BF) encoding technique, which is widely used in both research and practical applications of PPRL~\cite{Vat13,Ran13,Bro19}. We also provide an improved solution for privacy-preservation in the multi-party context using counting Bloom filter (CBF) encoding~\cite{Vat16b}.

%
%BF encoding has been used as an efficient $mask(\cdot)$
%function
%in several PPRL solutions~\cite{Vat13}.
\begin{defi}[\textbf{BF encoding}]
A BF $b_i$ is a bit vector of length $l$ bits
where all bits are initially set to $0$. $k$ independent hash
functions, $h_1, \ldots, h_k$, each with range $1, \ldots l$, are
used to map each of the elements $s$ in a set $S$ into the BF by
setting the bit positions $h_j(s)$ with $1 \le j \le k$ to $1$. 
\end{defi}

\begin{figure}[!t]
  \centering
  \scalebox{1.0}[0.9]{\includegraphics[width=0.47\textwidth]
                      {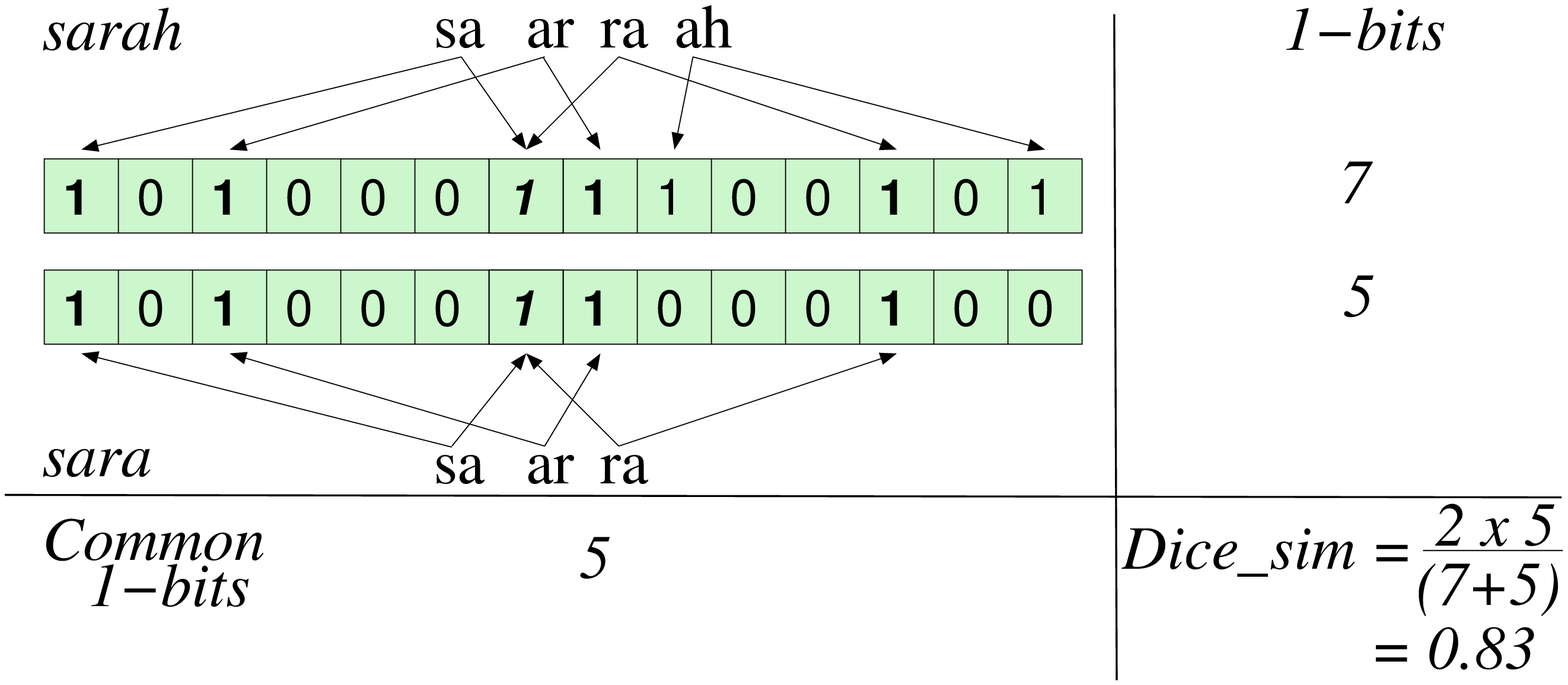}}
  \caption{\small{An example similarity (Dice coefficient) calculation of two strings masked using Bloom filter (BF) encoding,
                  where $l=14$, $k=2$, and $q=2$, as described in Section~\ref{sec-preli}.}}
           \label{fig:bloomfilter}
\end{figure}

For string matching, the $q$-grams (sub-strings of length
$q$) of QID values (that contain textual data, such as names and addresses)
of each record $r_{i,x}$ in the databases to be
linked $\mathbf{D}_i$, with $1 \le i \le p$, 
are hash-mapped into the BF $b_{i,x}$ using $k$ independent hash
functions~\cite{Sch15}.
Figure~\ref{fig:bloomfilter} illustrates the encoding of
bigrams ($q=2$) of two QID values `sarah' and `sara' into
$l=14$ bits long BFs using $k=2$ hash functions. The set of bigrams is first extracted from the string (e.g. \{'sa', 'ar', 'ra', 'ah'\} for 'sarah') and then each bigram in the set is hashed using $k=2$ hash functions to set the corresponding two indices in the BF to $1$ (e.g. hash-mapping bigram 'sa' results in setting the $1^{st}$ and $7^{th}$ bit positions to $1$). For numerical data, the neighbouring values (within a certain interval) of QID values are hash-mapped into the BF using $k$ hash functions~\cite{Vat16,Kara17}.
Collision of hash-mapping occurs (for example, the bigrams 'sa' and 'ra' are mapped to the same $7^{th}$ bit position in Figure~\ref{fig:bloomfilter}), which improves privacy of the encoding at the cost of loss in utility due to false positives.

%\subsection{Similarity calculation} 
%\label{subsec:sim_calc}

In order to allow fuzzy/approximate matching of masked QIDs to perform record linkage in the presence of typographical errors and variations, the similarity/distance between the encoded values needs to be calculated~\cite{Chr12,Vat13}.
%In order to allow approximate matching between records
%a similarity value of records needs to be calculated by comparing
%the masked QID values of the records.
The similarity of records masked into BFs
can be calculated either distributively across all database
owners~\cite{Vat14c,Vat12} or by a linkage unit~\cite{Dur13,Sch15}. %Vat12
Any set-based similarity function
(such as overlap, Jaccard, and Dice coefficient)~\cite{Chr12} can be used to calculate the
similarity of pairs or sets (multiple) of BFs. In PPRL, the Dice coefficient has
been used for matching of BFs since it is insensitive to
many matching zeros (bit positions to which 
no elements are hash-mapped) in long BFs~\cite{Sch15}.

\begin{defi}[\textbf{Dice coefficient similarity}]
The Dice coefficient similarity of $p$ ($p \ge 2$) BFs ($b_1, \cdots,
b_p$) is:
\begin{eqnarray}
\label{eq:Dice_coefficient}
sim (b_1, \cdots, b_p) &=& \frac{p \times z}{\sum_{i=1}^{p} x_i}, 
\end{eqnarray}
where $z$ is the number of common bit positions that are set to $1$ in
all $p$ BFs (common $1$-bits), and $x_i$ is the number
of bit positions set to $1$ in $b_i$ ($1$-bits), $1 \le i \le p$.
\end{defi}

For the example Bloom filter pair shown in
Figure~\ref{fig:bloomfilter}, the number of common $1$-bits is $5$ and the number of $1$-bits in the two Bloom filters are $7$ and $5$, respectively, and therefore the Dice coefficient similarity is calculated as $2 \times 5 / (7 + 5) = 0.83$.

\begin{defi}[\textbf{CBF encoding}]
A counting Bloom filter (CBF) $c$ is an integer vector of length $l$ bits that contains the counts of values in each bit position. Multiple BFs can be summarized into a single CBF $c$, such that $c[\beta] = \sum_{i=1}^p b_i[\beta]$, where $\beta, 1 \le \beta \le l$.  
$c[\beta]$ is the count value in the $\beta$ bit position of the CBF $c$ and
$b_i[\beta] \in[0,1]$ provides the value in the bit position $\beta$ of BF $b_i$.
Given $p$ BFs (bit vectors) $b_i$ with $1 \le i \le p$, the CBF $c$ can be generated by applying a vector addition operation between the bit vectors such that $c = \sum_i b_i$.
\end{defi}

\begin{theorem}
The Dice coefficient similarity of $p$ BFs can be calculated given only their corresponding CBF as:
\begin{eqnarray}
\label{eq:Dice_coefficient_cbf} 
sim (c) = \frac{p \times |\{\beta: c[\beta] = p, 1 \le \beta \le l\}|}{\sum_{\beta=1}^{l} c[\beta]} \nonumber
\\
\end{eqnarray}
\end{theorem}

\begin{proof}
The Dice coefficient similarity of $p$ BFs ($b_1$,$b_2$, $\cdots$, $b_p$) is determined by the sum of $1$-bits ($\sum_{i=1}^p x_i$) in the denominator of Eq.~(\ref{eq:Dice_coefficient}) and the number of common $1$-bits ($z$) in all $p$ BFs in the nominator of Eq.~(\ref{eq:Dice_coefficient}).
The number of $1$-bits in a BF $b_i$ is $x_i = b_i[1] + b_i[2] + \cdots + b_i[l]$, with $1 \le i \le p$. The sum of $1$-bits in all $p$ BFs is therefore $\sum_{i=1}^p x_i = \sum_{i=1}^p b_i[1] + b_i[2] + \cdots + b_i[l]$. The value in a bit position $\beta$ ($1 \le \beta \le l$) of the CBF of these $p$ BFs is $c[\beta] = b_1[\beta] + b_2[\beta] + \cdots + b_p[\beta]$. The sum of values in all bit positions of the CBF is $\sum_{\beta=1}^{l} c[\beta] = \sum_{\beta=1}^{l} b_1[\beta] + b_2[\beta] + \cdots + b_p[\beta]$ which is equal to $\sum_{i=1}^p x_i = \sum_{i=1}^p b_i[1] + b_i[2] + \cdots + b_i[l]$. 
Further, if a bit position $\beta$ ($1 \le \beta \le l$) contains $1$ in all $p$ BFs, i.e.\ $\forall_{i=1}^p b_i[\beta] = 1$, then $c[\beta] = \sum_{i=1}^p b_i[\beta] = p$. Therefore, the common $1$-bits ($z$) that occur in all $p$ BFs can be calculated by counting the number of positions $\beta \in c$ where $c[\beta] = p$, while the sum of the number of $1$-bits ($\sum_{i=1}^p x_i$) is calculated by summing the values in all bit positions $\beta \in c$, $\sum_{\beta=1}^l c[\beta]$.
\end{proof}

\begin{figure}[!t]
  \centering
  \scalebox{1.0}[1.0]{\includegraphics[width=0.47\textwidth]
                      {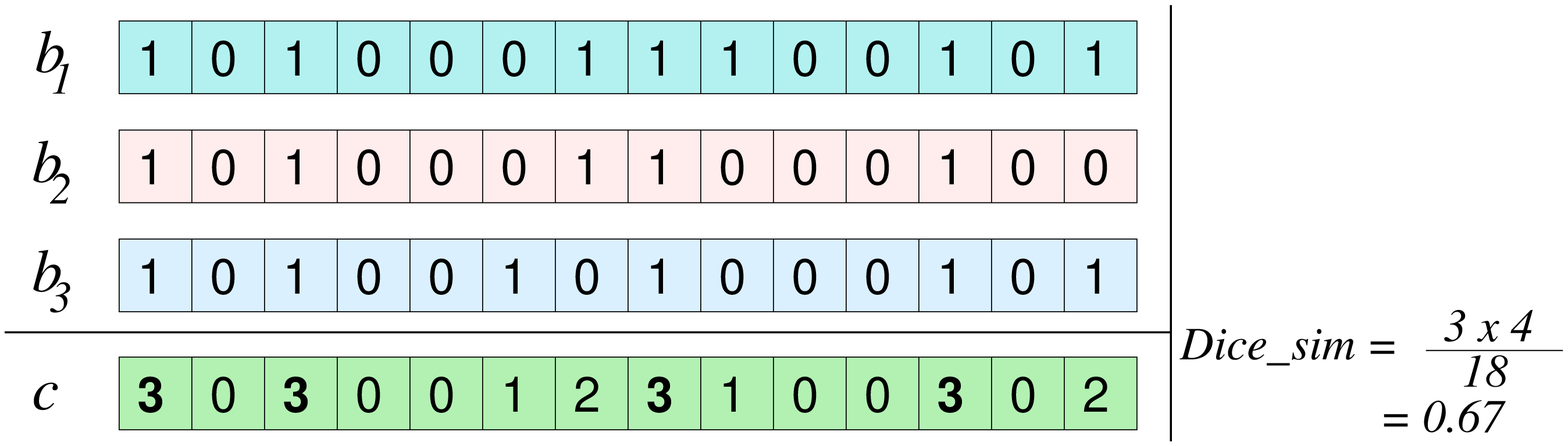}}
  \caption{\small{An example similarity (Dice coefficient) calculation of three BFs using their CBF, as described in Section~\ref{sec-preli}.}}
           \label{fig:cbloomfilter}
\end{figure}

Figure~\ref{fig:cbloomfilter} shows an example of using CBF to calculate the similarity of $p=3$ BFs ($b_1$, $b_2$, and $b_3$). The CBF $c$ contains the aggregated counts from the three BFs. The number of common $1$-bits in all three BFs is $4$ because $4$ indices in $c$ contain the count of $3$, and the total number of $1$-bits in all three BFs is $18$, which is the sum of the counts in $c$. Hence, the Dice coefficient similarity is calculated as $3 \times 4 / 18 = 0.67$.
As will be described in Sections~\ref{sec-privacy_improvement} and~\ref{subsec_privacy_analsis}, CBFs provide improved privacy compared to BFs in a multi-party context~\cite{Vat16b}.

%\subsection{Cluster graph representation and one-to-one matching} 
%\label{subsec:graph}

%We employ a graph-based matching approach at the $LU$ in our MP-PPRL protocol.
%%A graph $\mathbf{G}$ is formally defined as:
%\begin{defi}
%A graph $\mathbf{G}$ is an undirected graph $\mathbf{G} = (V,E)$, where $V$ is a set
%of vertices (vertices) and $E$ is a set of unordered pairs of vertices
%each representing an undirected edge between a pair of vertices.
%The vertex $v \in V$ 
%contains a masked record or a set of
%(merged) masked records from different parties
%(i.e.\ only one record from each party),
%and an edge $e \in E$ 
%represents the similarity $sim(v_i,v_j)$ between masked records in 
%the two vertices $v_i$ and $v_j$.
%%if $sim(v_i,v_j) \ge s_t$ , where
%%$s_t$ is the minimum similarity
%%threshold.
%\end{defi}

%\begin{defi}
%A cluster $c$ is a node/vertice in graph $\mathbf{G}$ 
%that contains at least two (masked) records and
%at maximum $p$ (masked) records:
%$c = v \in \mathbf{G}: if 2 \le len(v) \le p$,
%where $len(\cdot)$ returns the size 
%of the vertice $v$ (i.e.\ the number of records in $v$).
%\end{defi}

\section{MP-PPRL Protocol}
\label{sec-protocol}

Our protocol allows the efficient
identification of matching records from 
several (two or more) databases
held by different parties. 
We use an incremental graph-based clustering
approach to achieve
efficient linking of multiple large databases by reducing the exponential comparison space required by traditional linkage methods~\cite{Vat13,Chr11}. The explosion in the number of record pair comparisons required with increasing number of large databases necessitates a transition from batch to incremental clustering methods, which process one database at a time and typically store only a small subset of the data as potential matching records~\cite{Ack14}.
\smallskip

%\noindent
\textbf{Overview: } Masked/encoded database records are represented by the vertices in a graph and the similarities between compared records are represented by the edges. As we describe below, the databases are ordered using an \emph{ordering function} to determine in which order the databases are to be processed for incremental clustering. The aim of incremental clustering is to incrementally cluster/group vertices such that similar records from different databases are grouped into one cluster. Vertices containing similar records are identified by using a \emph{similarity function}. As we describe in Sections~\ref{subsec:early_mapping} and~\ref{subsec:late_mapping}, we propose two \emph{mapping functions} that perform clustering by merging and/or splitting vertices in the graph. The final output of our protocol is a cluster graph whose vertices are clusters containing similar records, or vertices containing a single record that is not matched with any other records. Each cluster/vertex in the final cluster graph corresponds to one real-world entity. Records within each cluster can be linked as matches and used for further analysis. In the following we describe our protocol in detail.

%A graph is an efficient data structure for 
%representing and processing similarities / relationships between records~\cite{Chr12,Far16}.
%%Graph-based matching has been used in the literature 
%%as an efficient and advanced
%%matching technique for record linkage~\cite{Chr12,Far16}.
%We use a cluster graph as a basis for identifying matches:

\begin{defi}[\textbf{Cluster graph}]
\label{def:graph}
A cluster graph $\mathbf{G}$ is a $p$-partite graph that contains a set $\mathbf{S}$ of non-empty independent sets $V_i$ with $1 \le i \le p$ containing vertices/nodes, and a set $E$ of unordered pairs of vertices each representing an undirected edge between a pair of vertices $v_x$ and $v_y$ such that $v_x \in V_i$ and $v_y \in V_j$ with $i \neq j$. %$\mathbf{G} = (V,E)$, where $V$ is a set
%of vertices and $E$ is a set of unordered pairs of vertices %(nodes)
%each representing an undirected edge between a pair of vertices.
The vertex $v$ can be considered as a \textbf{cluster}
containing either a single masked record (singleton) or a set of
masked records after merging vertices. %from different parties
%(i.e.\ only one record from each party),
An edge $e \in E$ 
represents the similarity $sim(v_x,v_y)$ between masked records in 
the two vertices $v_x$ and $v_y$. 
%if $sim(v_i,v_j) \ge s_t$ , where
%$s_t$ is the minimum similarity
%threshold.
\end{defi}

%\begin{defi}[\textbf{Cluster}]
%A \emph{cluster} $c$ is a vertex in the cluster graph $\mathbf{G}$ 
%that contains at least one masked record and
%a maximum of $p$ masked records:
%$c = v \in \mathbf{G}: 1 \le |v| \le p$,
%where $|v|$ returns the size 
%of the vertex $v$ (i.e.\ the number of records in $v$).
%A vertex containing only one masked record is also
%known as \emph{singleton}.
%\end{defi}

Following Definition~\ref{def:graph}, the records from all databases $\mathbf{D}_1, \cdots, \mathbf{D}_p$ are
represented as vertices in a cluster graph $\mathbf{G}$, 
and they are incrementally clustered  
such that at the end of our protocol 
each cluster contains a set of matching records from different parties.
During  incremental clustering we have to assign records of a newly considered party 
to the already determined clusters of previously matched parties. 
In general, new records might be similar to several such clusters so that there 
is a many-to-many match relationship between the set X of already existing clusters
and the set Y of new records as shown in Figure~\ref{fig:hungarian} (left-hand side).
Our goal, however, is to 
identify the best one-to-one mapping for such matches (Figure~\ref{fig:hungarian} (right-hand side)) 
since the databases are assumed to be deduplicated, and therefore only one-to-one true mapping can exist between records from different databases. 

\begin{figure}[t!]
\centering
\scalebox{0.9}[0.9]{\includegraphics[width=0.45\textwidth]{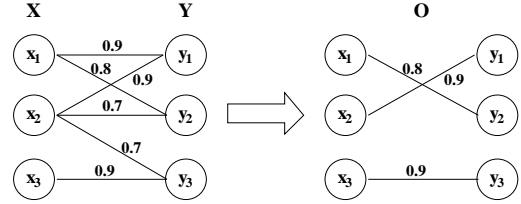}}
\caption{\small{An example of optimal one-to-one mapping (defined in Section~\ref{sec-protocol}) using the Hungarian algorithm~\cite{Kuh55}.}
}
\label{fig:hungarian}
\end{figure}

Such one-to-one mappings between the vertices in $\mathbf{G}$
can be determined by either (1) a greedy approach
or (2) an optimal mapping approach that ensures that each record (vertex)
is matched with only the best matching record/records 
from other parties.
Given two lists of (unassigned) vertices $X$ and $Y$,
the greedy approach scans through the vertices in $X$ and
assigns them to the best matching vertex in $Y$
that is not yet 
assigned to any other vertex
according to their similarity.
The greedy approach is not optimal, because when assigning a vertex $x \in X$ to a vertex in $Y$ only the
similarities between $x$ and unassigned vertices in $Y$ are considered
while neglecting the similarities of the other vertices in $X$ with vertices in $Y$. Moreover, similar to the best link grouping method proposed by Kendrick~\cite{Ken98} (as described in Section~\ref{sec-related}), greedy mapping depends on the ordering of the vertices/nodes as they are processed.

In our protocol, we therefore use the optimal mapping
approach using the Hungarian algorithm~\cite{Kuh55}, 
which is a combinatorial
algorithm for solving the optimal assignment problem in polynomial time.
Given two sets of vertices, $X$ and $Y$,
the algorithm determines the optimal one-to-one mapping 
by assigning a vertex in $X$ to a maximum of one vertex in $Y$ such that the overall similarity
between all assigned vertices is maximized:  
 
\begin{defi}[\textbf{Optimal mapping}]
\label{def:optmap}
Given two sets of vertices, $X$ and $Y$, 
along with a similarity function
$sim(x_i \in X, y_j \in Y)$. Identify a bijection $O: X \rightarrow Y$ such that
\begin{eqnarray}
\label{eq:cost_func}
\sum_{x_i,y_j \in O} sim(x_i, y_j)
\end{eqnarray}
is maximized.
\end{defi}

An illustrative example of optimal one-to-one mapping
is shown in Figure~\ref{fig:hungarian}.
For example, $x_1$ has the highest similarity with $y_1$ of $sim(x_1,y_1)=0.9$ while with $y_2$ the similarity is $sim(x_1,y_2) = 0.8$. With greedy mapping (assuming the order of processing as $x_1$ first and then $x_2$ followed by $x_3$), $x_1$ is mapped to $y_1$ and therefore $x_2$ needs to be mapped to $y_2$, and $x_3$ with $y_3$. This gives a total summed similarity of $2.5$ (Eq.~\ref{eq:cost_func}). However, with the optimal one-to-one mapping, $x_1$ is mapped with $y_2$ and $x_2$ with $y_1$ while $x_3$ is still mapped to $y_3$, resulting in a total similarity value of $2.6$ (which is better than the greedy mapping).
If $|X| \neq |Y|$, $abs(|X|-|Y|)$ vertices remain not mapped
to any vertices after one-to-one mapping is applied.

The proof of Kuhn-Munkres theorem states that for any matching $O$ and any feasible labelling $O'$ (such as greedy mapping), it holds~\cite{Kuh55}
\begin{eqnarray}
\label{eq:munkres}
w(O) = \sum_{e \in O} w(e) >= \sum_{e' \in O'} w(e'),
\end{eqnarray}
where $w(\cdot)$ denotes the weight function of an edge.
Therefore, $O$ is the optimal mapping in terms of maximizing edge weights (similarities in our context).
We will experimentally evaluate the greedy as well as the optimal mapping approaches in Section~\ref{sec-experiment}. 
%
%\subsection{Protocol Steps}
%\label{sec-steps}
%
%Our protocol allows efficient, approximate, and private linking of
%records based on their masked QID values in multiple databases
%from $P$ $(\ge 2)$ sources/parties.

We three initial steps of our MP-PPRL protocol are: %to achieve such efficient, approximate, and
%We illustrate the steps using 
%three example datasets held by three parties, 
%as shown in Figure~\ref{fig-label_mpbf1}. 
%The main steps of our protocol are:

\begin{figure*}[t!]
\centering
\scalebox{0.9}[0.85]{\includegraphics[width=0.98\textwidth]{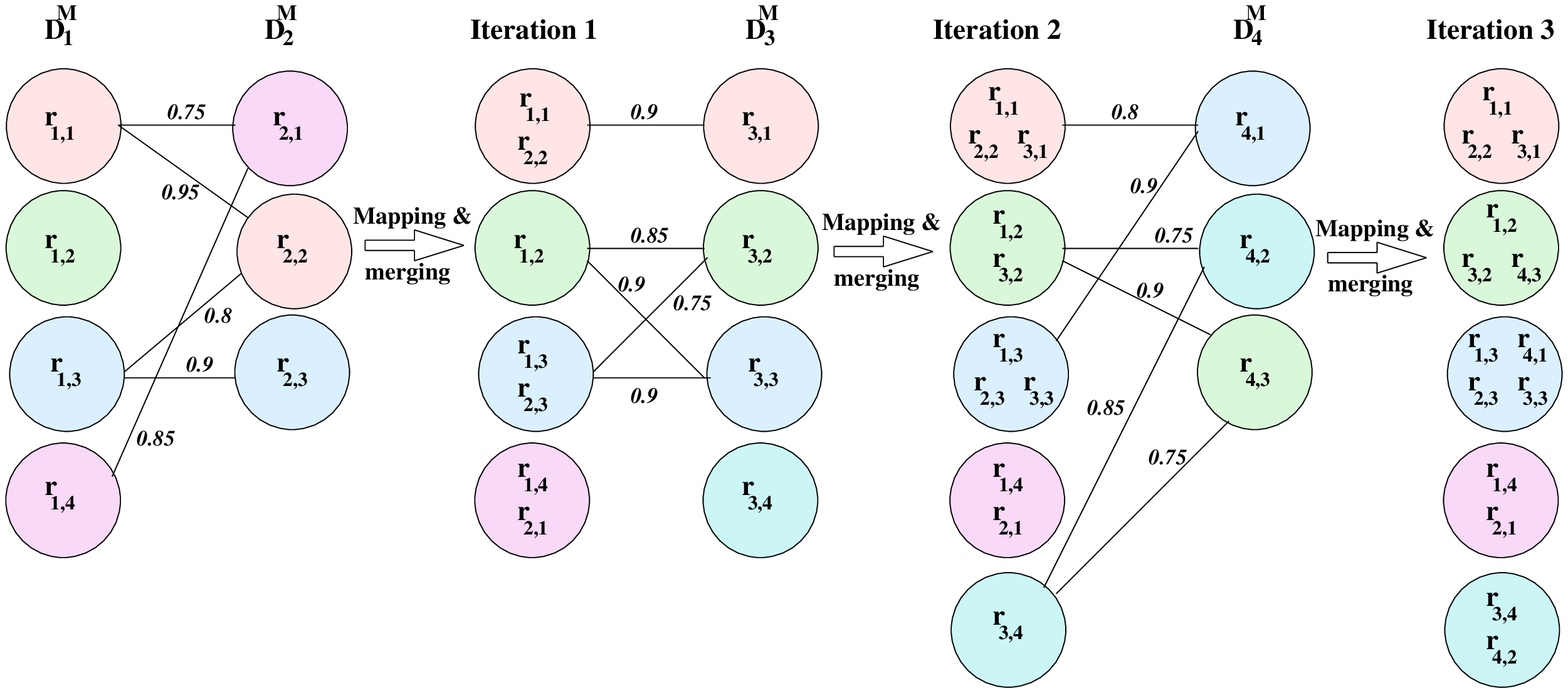}}
\caption{\small{An example of early mapping-based incremental clustering, as described in Section~\ref{subsec:early_mapping}. Edges represent a similarity value between vertices of at least the similarity threshold $s_t$ ($s_t = 0.75$ in this example). 
%Assuming databases are individually deduplicated, one-to-one optimal mapping is performed at every iteration (early mapping) to prune (unlikely matching) edges before merging the matched vertices connected by an edge, as described in Section~\ref{subsec:early_mapping}.
The similarity values shown here are made-up example values.
Different colors represent different final clusters and how they are iteratively mapped and merged.}
}
\label{fig:clus_mpprl}
\end{figure*}

%\begin{figure*}[t!]
%\centering
%\scalebox{0.9}[0.85]{\includegraphics[width=0.98\textwidth]{clus_mpprl_em}}
%\caption{\small{An example of early mapping-based incremental clustering, as described in Section~\ref{subsec:early_mapping}. Records are first inserted into the graph as separate vertices, and then matching vertices (records) are iteratively merged. The number of iterations in this example is $p-1 = 3$ for linking $p=4$ databases. Edges represent a similarity value between vertices of at least the similarity threshold $s_t$ (with $s_t = 0.75$ in this example). 
%Assuming databases are individually deduplicated, one-to-one optimal mapping is performed at every iteration (early mapping) to prune (unlikely matching) edges before merging the matched vertices connected by an edge, as described in Section~\ref{subsec:early_mapping}.
%Different colors represent different final clusters and how they are iteratively mapped and merged.}
%}
%\label{fig:clus_mpprl}
%\end{figure*}

\begin{enumerate}[(1)]

\item All database owners mask (encode) their database records using the same
masking function $mask(\cdot)$.
This can, for example, be BF encoding,  
as described in Section~\ref{sec-preli}.

\item To reduce the comparison space, a blocking function 
$block(\cdot)$ is applied on the
database records (individually by the database owners) 
to group similar records into the same block according to some
criteria (known as blocking key)~\cite{Chr12,Ran14}. 
All records that have the same
(or a similar) blocking key
value (BKV) are grouped into the same block.
For example, phonetic-based 
or multi-bit tree-based blocking can be used as the
$block(\cdot)$ function~\cite{Chr12,Ran14,Sch15}.

\item The masked records ($\mathbf{D}_i^M$) along with their blocks ($\mathbf{B}_i$)
are sent to a linkage unit ($LU$) to conduct the linkage of
these masked records using 
the graph-based incremental clustering approach.
At the $LU$, the records are processed block by block
(i.e. each block $B \in \mathbf{B}$ 
is considered as one graph $\mathbf{G}_{B}$,
where $\mathbf{B}$ contains the union of all $\mathbf{B}_i$, with $1 \le i \le p$).

\end{enumerate}

%As will be described below, we propose two methods for applying one-to-one mapping between
%vertices in the graph: (1) early mapping and (2) late mapping.
%In Section~\ref{subsec:early_mapping} we first present our MP-PPRL
%approach using early mapping and then in Section~\ref{subsec:late_mapping}
%we describe the differences of the late mapping-based approach.

%As will be described below, 
We propose two different
methods for incremental clustering in the graphs $\mathbf{G}_B$: 
(1) early mapping and (2) late mapping.
We first present the steps involved in 
the incremental clustering
approach with
early mapping in Section~\ref{subsec:early_mapping} 
and then the late
mapping-based approach
in Section~\ref{subsec:late_mapping}. While both approaches incrementally merge records from different parties,
they differ in when they apply the one-to-one mapping restriction. 
With early mapping this restriction is continually observed  
such that every record is only assigned to a single cluster and
the number of records per cluster never exceeds the number of parties. 
By contrast, late mapping assigns records to all clusters for which a minimum 
similarity is exceeded so that there may temporarily be overlapping clusters and clusters with several records from the
same party. The one-to-one restriction is then enforced at the end
of the algorithm in a separate mapping phase.  
Both approaches have a 
trade-off between complexity and linkage quality,
as we will discuss in Section~\ref{sec-analysis}.

As will be detailed in the following two sections,
the inputs to the incremental clustering algorithm are:
$p$ masked databases $\mathbf{D}_i^M$ (with $1 \le i \le p$), 
the union of blocks from all parties $\mathbf{B} = \cup_i \mathbf{B}_i$, a similarity
function $sim(\cdot)$ for calculating similarities between
vertices in $\mathbf{G}_{B}$, an ordering function $ord(\cdot)$
for ordering the databases to be processed, a mapping function $map(\cdot)$
for one-to-one mapping between vertices in $\mathbf{G}_{B}$ (early mapping, late mapping, or the na\"ive greedy mapping), 
a minimum similarity threshold $s_t$ to connect two vertices in $\mathbf{G}_{B}$
by an edge (if their similarity is at least $s_t$), and the minimum subset size $s_m$ ($s_m \le p$), i.e.\ the minimum number of records that each 
%matching record sets (
final cluster must contain.

The databases need to be ordered using the $ord(\cdot)$ function for incremental clustering.
The ordering can be
either (a) random, (b) according to their
sizes in descending order so that a smaller number of
merging will be required,
or (c) depending on their data quality of the respective
databases 
in descending order so that the initial clusters will be of higher quality leading to higher linkage quality~\cite{Nen18}.

\subsection{Early mapping-based clustering} % approach}
\label{subsec:early_mapping}

The early mapping-based clustering 
incrementally adds records in each database
to the corresponding vertices in the graph by
identifying the one-to-one mapping between vertices
and records and then merging them.
To achieve the one-to-one mapping we  apply the Hungarian algorithm~\cite{Kuh55} according to Definition~\ref{def:optmap} ensuring that a record from one database
is matched to a maximum of one cluster of previously matched records and that  clusters in the graph are non-overlapping 
(i.e.\ $\forall (v_i,v_j) \in \mathbf{G}_B ~:~ v_i \cap v_j = \emptyset$).

%involves two main phases:
%(1) merging vertices and (2) one-to-one mapping
%between vertices. 

\smallskip
Selecting the optimal cluster to which a record should be added is based on the similarities between  two vertices of the cluster graph and a minimum
similarity threshold $s_t$. In other words, two vertices $v_i,v_j \in \mathbf{G}_B$
are only merged into one if $sim(v_i,v_j) \ge s_t$.
The similarity between two singletons can easily
be calculated using a similarity function,
for example the Dice coefficient similarity, to
compare the singletons containing records
masked into BFs (as described in Section~\ref{sec-preli}).
%
%We use the average similarity to calculate 

The similarity between a cluster $c$ that contains more than one record and a singleton $v$ consisting of a single
masked record can be calculated in several ways, including maximum similarity (single linkage), minimum similarity (complete linkage), or average similarity. We use the average similarity function in this work 
in order to consider data errors and variations, as well as possible variations of the masking function (such as Bloom filter collisions~\cite{Sch09}) while not compromising computational efficiency. We leave studying other similarity functions for our incremental clustering as a future work.

\begin{defi}[\textbf{Average similarity}]
The average similarity between a cluster $c$ and a 
(masked) record $r_{i,x}$ (in a singleton) is 
\begin{small}
\begin{eqnarray}
\label{eq:avg_sim}
sim_{avg}(c,r_{i,x}) &=& \frac{\sum_{r_{j,y} \in c} sim(r_{j,y},r_{i,x})}{|c|},
\end{eqnarray}
\end{small}
with $1 \le i,j \le p$, $i \neq j$, and $|c| \ge 1$.
\end{defi}

The early mapping-based approach involves $p-1$ iterations to 
perform one-to-one mapping and merging
between records from $p$ parties.
An overview of our clustering approach 
with early one-to-one mapping is illustrated for linking $p=4$ databases ($p-1=3$ iterations) in Figure~\ref{fig:clus_mpprl}
and outlined in Algorithm~1. 
The steps of our protocol with early mapping-based
clustering are (continuing after the initial steps (1) to (3)):

\begin{enumerate}[(1)]

\setcounter{enumi}{3}

\item The $LU$ conducts linkage of masked records
in databases $\mathbf{D}_1^M, \cdots, \mathbf{D}_p^M$ 
from $p$ parties.
These databases are ordered (using the $ord(\cdot)$ function in line~2 in 
Algorithm~1) for incremental clustering.
%The ordering can be
%either (a) random, (b) according to their
%sizes in descending order so that a smaller number of
%merging will be required,
%or (c) depending on their data quality %of the respective
%%databases 
%in descending order so that the initial clusters will be of higher quality leading %to higher linkage quality.
%
For each block $B \in \mathbf{B}$, 
the masked records in $B$ of the first party
$P_1$ are
added into a graph $\mathbf{G}_{B}$ as separate vertices 
(lines~3 to 9 in Algorithm~1).
%Note that the $ord$ function may influence
%the linkage quality depending on the data quality.
%influences the matching results. 
%In Section~\ref{sec-experiment}, we empirically evaluate different 
%ordering of parties (sorted, reverse, and random) for matching.
%After adding the first party's records into $\mathbf{G}_B$, 
%it contains $n_i/b$ vertices,
%assuming the size of $\mathbf{D}_i$ is $n_i$ and $b$ blocks of
%equal size are generated by the $block(\cdot)$ function
%(i.e.\ the number of
%records in $\mathbf{D}_i$ that belong to block $B$ is $|B| = n_i/b$).
Then the second party $P_2$'s 
masked records are inserted
into $\mathbf{G}_{B}$ as separate vertices and the similarities between vertices
of the first party and the second party are calculated (lines~10 to 13). 
If the similarity
between two vertices
is at least a minimum threshold $s_t$ an edge is created between the
corresponding vertices (as shown in Figure~\ref{fig:clus_mpprl} for $s_t=0.75$ and
described in lines~14 and 15 in Algorithm~1).
%the two vertices are merged into a cluster in $\mathbf{G}_B$
%(as shown in Figure~\ref{fig:clus_mpprl} and
%explained in lines~13-14 in Algorithm~1).

\item The optimal
one-to-one mapping (as defined in Section~\ref{sec-preli}) 
is applied
in every iteration $i$ (with $1 \le i \le p-1$)
%using the $map(\cdot)$ function 
after edges between the records from 
party $P_{i+1}$ (singletons)
and clusters of records from parties $P_1$ to $P_{i}$ have been added.
%This is conducted based on the calculated similarities
%between the connected vertices (by an edge) in $\mathbf{G}_B$ 
%in order to
%prune edges (lines~16-20). 
This optimal mapping connects only two highly matching
vertices, complying with the assumption of
deduplication. 
All the edges $e \in \mathbf{G}_{B}.E$ 
that are not matching after the
optimal mapping are removed from $\mathbf{G}_{B}$ %the graph
%along with their corresponding vertices 
(lines~16 to 19). 
%It is important to note that early one-to-one mapping reduces
%the number of clusters for comparisons and therefore is more efficient,
%but it might reduce the quality of the final linkage results. 
%%If late one-to-one mapping is used, this step is not required.

%can be conducted at
%any iteration $opt_i$ (except the first iteration) after adding $P_i$th
%records, with $2 \le i \le p$, into the graph $\mathbf{G}_{B}$.
%Early one-to-one matching with $opt_i$ staring from a small value
%(for example, $opt_i = [i: i \ge 2]$ conducts one-to-one matching
%after every pair of database records are matched) is more efficient,%
%but it reduces the 
%recall of linkage. Late one-to-one matching (for example,
%$opt_i = [i: i \ge p]$ allows one-to-one matching after $p-1$ database records
%are matched with the last database records) improves the recall
%at the cost of more comparisons (due to more intermediate matching vertices).
%In Section~\ref{sec-experiment}, we empirically evaluate and
%compare early and late one-to-one matching approaches in terms
%of linkage quality and scalability.
%If late one-to-one mapping is used, this step is not required.

\begin{table}[t]
\scriptsize\addtolength{\tabcolsep}{-3pt}
\begin{tabular}{lll} %{0.47\textwidth}{c @{\extracolsep{\fill}} lll}
  \label{algo_early_map}
    ~ \\[0.5mm] \hline 
    \\[-2mm]
    \multicolumn{3}{l}{\textbf{Algorithm~1:} Early mapping-based incremental clustering (Section~\ref{subsec:early_mapping})}
      \\[0.5mm] \hline
    ~ \\[-2mm]
    \multicolumn{3}{l}{\textbf{Input:}} \\
    {- $\mathbf{D}^M_i$} & \multicolumn{2}{l}{: Party $P_i$'s BFs along with their BKVs, $1 \le i \le p$} \\
    {- $\mathbf{B}$} & \multicolumn{2}{l}{: Blocks containing the union of blocks from all parties} \\
    {- $sim$} & \multicolumn{2}{l}{: Similarity function} \\
    {- $ord$} & \multicolumn{2}{l}{: Ordering function for incremental processing of databases} \\
    {- $map$} & \multicolumn{2}{l}{: One-to-one mapping function} \\
    {- $s_t$} & \multicolumn{2}{l}{: Minimum similarity threshold to classify record sets} \\
    {- $s_m$} & \multicolumn{2}{l}{: Minimum subset size, with $2 \le s_m \le p$} \\
    \multicolumn{3}{l}{\textbf{Output:}} \\
    {- $\mathbf{M}$} & \multicolumn{2}{l}{: Matching record sets (clusters)} \\[1mm]

    1:& $clus\_ID = 0; \mathbf{G} = \{\}; \mathbf{M} = \{\}$ & // Initialization \\
    2:& $DBs = ord([\mathbf{D}^M_1,\mathbf{D}^M_2, \cdots, \mathbf{D}^M_p])$ & // Order databases \\
    3:& \textbf{for} $B \in \mathbf{B}$ \textbf{do}: & // Iterate blocks \\
    4:& \hspace{2mm} $\mathbf{G}_B = \{\}$ & // Graph for block $B$ \\
    5:& \hspace{2mm} \textbf{for} $i \in [1, 2, \cdots p]$ \textbf{do}:& // Iterate parties \\
    6:& \hspace{4mm} \textbf{if} $i == 1$ \textbf{do}: & // First party\\    
    7:& \hspace{6mm} \textbf{for} $rec \in DBs[i]$ \textbf{do}: & // Iterate records \\
    8:& \hspace{8mm} $clus\_ID += 1$ & ~ \\
    9:& \hspace{8mm} $\mathbf{G}_{B}[clus\_ID] = [DBs[i][rec]]$ & // Add vertices \\
    10:& \hspace{4mm} \textbf{if} $i > 1$ \textbf{do}: & // Other parties\\ 
    11:& \hspace{6mm} \textbf{for} $rec \in DBs[i]$ \textbf{do}: & // Iterate records \\
    12:& \hspace{8mm} \textbf{for} $c \in \mathbf{G}_B$ \textbf{do}: & // Iterate vertices \\
    13:& \hspace{10mm} $sim\_val = sim(rec,c)$ & // Calculate similarity \\
    14:& \hspace{10mm} \textbf{if} $sim\_val \ge s_t$ \textbf{then}:&~\\
    15:& \hspace{12mm} $\mathbf{G}_B.add\_edge(c,rec)$ & // Add edges \\
    16:& \hspace{6mm} $opt\_E = map(\mathbf{G}_B.E)$ & // 1-to-1 mapping \\
    17:& \hspace{6mm} \textbf{for} $e \in \mathbf{G}_B.E$ \textbf{do}: & // Iterate edges \\
    18:& \hspace{8mm} \textbf{if} $e \notin opt\_E$ \textbf{then}: & ~ \\
    19:& \hspace{10mm} $\mathbf{G}_B.remove(e)$ & // Prune edges\\ 
    20:& \hspace{4mm} \textbf{for} $e \in \mathbf{G}_B.E$ \textbf{do}: & // Remaining edges \\
    21:& \hspace{6mm} $\mathbf{G}_B.merge(get\_vertices(e))$ & // Merge cluster vertices \\
    22:& \hspace{2mm} $\mathbf{G}.add(\mathbf{G}_B)$ & // Add $B$'s clusters to $\mathbf{G}$\\ 
    23:& \hspace{0mm} \textbf{for} $c \in \mathbf{G}$ \textbf{do}: & // Iterate final clusters\\
    24:& \hspace{2mm} \textbf{if} $|c| \ge s_m$ \textbf{then}: & // Size at least $s_m$ \\
    25:& \hspace{4mm} $\mathbf{M}.add(c)$ & // Add to $\mathbf{M}$\\ 
    26:& return $\mathbf{M}$ & // Output $\mathbf{M}$ \\ %[1mm]
      \hline %\\ %[2mm]
  \end{tabular}
\end{table}

\item The vertices that are connected by an edge are then merged
into one (lines~20 and 21), while the vertices that do not have
any connecting edge (those that are not matching to any vertices
in the other databases) are kept as unclustered vertices.

In our running example shown in Figure~\ref{fig:clus_mpprl}, 
the optimal mapping (according to the objective function~\ref{eq:cost_func}) between records of parties $P_1$ and $P_2$ (based on the similarity values) in the first iteration leads to their respective
records $r_{1,1}$ and $r_{2,2}$
to be merged into a single cluster, while $r_{1,2}$ from $P_1$
is not clustered
with any vertices from $P_2$. Similarly, $r_{1,3}$ and $r_{2,3}$, and
$r_{1,4}$ and $r_{2,1}$ are merged into clusters.

\item The $LU$ then proceeds with the masked records in $B$ of the
next (third) party which are first inserted into $\mathbf{G}_B$ as separate
vertices (singletons). Then the similarities between these vertices and the
clustered vertices and singletons from the previous parties' masked
records are calculated and an edge is created connecting those vertices
that have a similarity above the minimum threshold $s_t$ (lines~10 to 15
in Algorithm~1).
%A basic approach for calculating similarities between a clustered
%vertex and a single vertex is by calculating the average of
%similarities between each record in the clustered vertex and the
%single vertex. 
%In Section~\ref{sec-sim} we describe different comparison
%and similarity functions.
%The similarity between a clustered vertex from previous
%parties and a singleton from the current party
%can be calculated using the similarity functions 
%described in Section~\ref{sec-protocol}.
%%The average similarity between a clustered vertex $clus$ and a new
%%vertex containing a single (masked) record $r_{i,x}$ is 
%%$sim_{avg}(clus,r_{i,x}) = \frac{\sum_{r_{j,y} \ in clus} sim(r_{j,y},r_{i,x})}{|clus|}$, while
%%the minimum similarity between $clus$ and $r_{i,x}$ is 
%%$sim_{min}(clus,r_{i,x}) = min_{r_{j,y} \in clus}(sim(r_{j,y},r_{i,x}))$,
%%with $1 \le i,j \le p$ and $i \neq j$.
An optimal one-to-one mapping is then applied again between
the vertices from all previous parties and the new 
singleton vertices of the current party
%clustered vertices and single new vertices if early one-to-one mapping is used
in lines~16 to 19.
For example in Figure~\ref{fig:clus_mpprl}, in iteration 2 the 
singleton vertex with record $r_{3,3}$ of the current party $P_3$
and the clustered vertex containing records $(r_{1,3},r_{2,3})$ of
previous parties $P_1$ and $P_2$, respectively,
are merged. Similarly, $r_{3,1}$ is merged with the cluster containing $r_{1,1}$ and $r_{2,2}$ from previous parties, while $r_{3,2}$ is merged with $r_{1,2}$, as this gives the optimal mapping (according to Equation~\ref{eq:cost_func}).

%\item If the similarity between a clustered vertex and a single vertex with a
%new record is less than $s_t$ (line~12), individual pair-wise similarities 
%are calculated to split the vertex into two (lines~13-20). The records in the clustered
%vertex that have a similarity with the new record in the single vertex above $s_t$
%are clustered into a new vertex with the new record, while the previous
%clustered vertex remains the same without merging the new record with it.
%An optimal one-to-one matching can be applied between
%clustered vertices and single new vertices if the current iteration $i \in opt_i$.

\item The vertices connected by an edge after one-to-one mapping
(highly matching vertices) at an iteration 
are merged into one, while the vertices (both clustered and  
singletons) that are not matching to any other vertices
remain as unclustered vertices (lines~20 and 21). For example, the 
vertex with record $r_{3,4}$ of party $P_3$
and the clustered vertex containing records $(r_{1,4},r_{2,1})$ of
parties $P_1$ and $P_2$, respectively,
are not merged with any other vertices in iteration 2, as shown in Figure~\ref{fig:clus_mpprl}.

\item This process of mapping and merging of vertices is repeated
until the masked records of all parties are processed 
(i.e. $p-1$ iterations for each block).
The output will be clusters
(final vertices in graph $\mathbf{G}_B$) that either have records from all $p$ parties,
or a subset of $p$ parties, or only one record from a single party.
The final clusters of block $B$ (i.e.\ vertices in graph $\mathbf{G}_B$)
are added to $\mathbf{G}$ (line~22).
Based on the minimum subset size $s_m$ required by the MP-PPRL protocol, 
all the vertices that have a size
of at least $s_m$ 
(i.e.\ vertices containing matching records from at least $s_m$ parties) 
are added to the final matching set of records $\mathbf{M}$ (lines~23 to 26). 
%The clusters of size less than $s_m$ (with $s_m \le p$) are pruned to identify
%matching records from at least $s_m$ parties.
For example, if $s_m = 3$ in our running example, then $\mathbf{M}$ will contain only three clusters which are $(r_{1,1},r_{2,2},r_{3,1})$, $(r_{1,2},r_{3,2},r_{4,3})$, and $(r_{1,3},r_{2,3},r_{3,3},r_{4,1})$.

\end{enumerate}

\begin{figure*}[t!]
\centering
\scalebox{1.0}[1.02]{\includegraphics[width=0.99\textwidth]{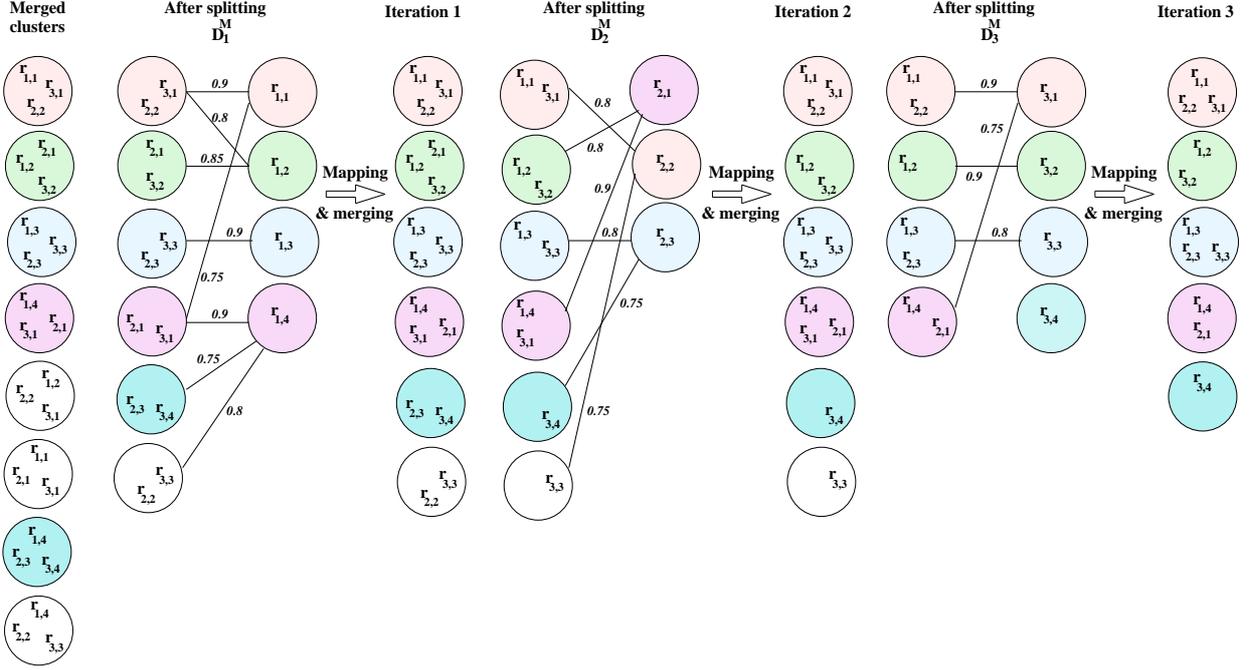}}
\caption{\small{An example of the splitting and one-to-one mapping phases of late mapping-based incremental clustering,
%for large-scale MP-PPRL. 
as described in Section~\ref{subsec:late_mapping}. The similarity values shown here are
made-up example values (with the similarity threshold $s_t = 0.75$).
Different colors represent different final clusters and how they are iteratively split from merged clusters and mapped.}
}
\label{fig:clus_mpprl2}
\end{figure*}

%\begin{figure*}[t!]
%\centering
%\scalebox{1.0}[1.02]{\includegraphics[width=0.99\textwidth]{clus_mpprl_lm2}}
%\caption{\small{An example of the splitting and one-to-one mapping phases of late mapping-based incremental clustering.
%%for large-scale MP-PPRL. 
%Late one-to-one mapping is performed after records from all databases are first matched (many-to-many) and merged in a cluster graph $\mathbf{G}_B$ for block $B$ in order to generate non-overlapping final clusters, as described in Section~\ref{subsec:late_mapping}. 
%Merged clusters are iteratively split for each database $\mathbf{D}_i$'s records (databases are ordered according to the $ord(\cdot)$ function) and one-to-one mapped with unique sets of records from other databases $\mathbf{D}_j$ in the merged clusters, with $1 \le i,j \le p$ and $i \neq j$ ($p=3$ in this example). The number of edges generated for mapping in each iteration corresponds to the number of clusters that appear before splitting in that iteration. The number of iterations is equal to $p$ ($p=3$) and $s_t = 0.75$ in this example.
%Different colors represent different final clusters and how they are iteratively split from merged clusters and mapped.}
%}
%\label{fig:clus_mpprl2}
%\end{figure*}

\begin{table}[!t]
\scriptsize\addtolength{\tabcolsep}{-3pt}
\begin{tabular*}{0.47\textwidth}{c @{\extracolsep{\fill}} lll}
  \label{algo_late_map}
    ~ \\[0.5mm] \hline 
    \\[-2mm]
    \multicolumn{3}{l}{\textbf{Algorithm~2:} Late mapping-based incremental clustering (Section~\ref{subsec:late_mapping})}
      \\[0.5mm] \hline
    ~ \\[-2mm]
    \multicolumn{3}{l}{\textbf{Input:}} \\
    {- $\mathbf{D}^M_i$} & \multicolumn{2}{l}{: Party $P_i$'s BFs along with their BKVs, $1 \le i \le p$} \\
    {- $\mathbf{B}$} & \multicolumn{2}{l}{: Blocks containing the union of blocks from all parties} \\
    {- $sim$} & \multicolumn{2}{l}{: Similarity function} \\
    {- $ord$} & \multicolumn{2}{l}{: Ordering function for incremental processing of databases} \\
    {- $map$} & \multicolumn{2}{l}{: One-to-one mapping function} \\
    {- $s_t$} & \multicolumn{2}{l}{: Minimum similarity threshold to classify record sets} \\
    {- $s_m$} & \multicolumn{2}{l}{: Minimum subset size, with $2 \le s_m \le p$} \\
    \multicolumn{3}{l}{\textbf{Output:}} \\
    {- $\mathbf{M}$} & \multicolumn{2}{l}{: Matching record sets (clusters)} \\[1mm]

    1:& $clus\_ID = 0; \mathbf{G} = \{\}; \mathbf{M} = \{\}$ & // Initialization \\
    2:& \textbf{for} $B \in \mathbf{B}$ \textbf{do}: & // Iterate blocks \\
    3:& \hspace{2mm} $\mathbf{G}_B = \{\}$ & // Graph for block $B$ \\
    4:& \hspace{2mm} \textbf{for} $i \in [1, 2, \cdots p]$ \textbf{do}:& // Iterate parties \\
    5:& \hspace{4mm} \textbf{if} $i == 1$ \textbf{do}: & // First party\\    
    6:& \hspace{6mm} \textbf{for} $rec \in \mathbf{D}^M_i$ \textbf{do}: & // Iterate records \\
    7:& \hspace{8mm} $clus\_ID += 1$ & ~ \\
    8:& \hspace{8mm} $\mathbf{G}_{B}[clus\_ID] = [\mathbf{D}^M_i[rec]]$ & // Add vertices \\
    9:& \hspace{4mm} \textbf{if} $i > 1$ \textbf{do}: & // Other parties\\ 
    10:& \hspace{6mm} \textbf{for} $rec \in \mathbf{D}^M_i$ \textbf{do}: & // Iterate records \\
    11:& \hspace{8mm} \textbf{for} $c \in \mathbf{G}_B$ \textbf{do}: & // Iterate vertices \\
    12:& \hspace{10mm} $sim\_val = sim(rec,c)$ & // Calculate similarity \\
    13:& \hspace{10mm} \textbf{if} $sim\_val \ge s_t$ \textbf{then}:&~\\
    14:& \hspace{12mm} $\mathbf{G}_B.add\_edge(c,rec)$ & // Add edges \\
    15:& \hspace{4mm} \textbf{for} $e \in \mathbf{G}_B.E$ \textbf{do}: & // Iterate edges \\
    16:& \hspace{6mm} $\mathbf{G}_B.merge(get\_vertices(e))$ & // Merge cluster vertices \\
    %17:& ~~ ~~\textbf{for} $v \in \mathbf{G}_B.V$ \textbf{do}: & // Iterate merged vertices\\
    %18:& ~~ ~~ ~~\textbf{if} $|[\mathbf{D}^M_i.recs \in v]| > 1$ \textbf{do}: & ~\\
    %19:& ~~ ~~ ~~ ~~$\mathbf{G}_B.split\_with\_copy(\mathbf{D}^M_i.recs)$ & // Split vertices\\
    17:& \hspace{2mm} $DBs = ord([\mathbf{D}^M_1,\mathbf{D}^M_2, \cdots, \mathbf{D}^M_p])$ & // Order databases \\
    18:& \hspace{2mm} \textbf{for} $i \in [1, 2, \cdots p]$ \textbf{do}:& // Iterate parties \\
    %22:& ~~ ~~$\mathbf{G}_B.split\_without\_copy(DBs[i].recs)$ & // Split this party's records \\
    19:& \hspace{4mm} $\mathbf{G}_B.split(DBs[i].recs)$ & // Split this party's records \\
    20:& \hspace{4mm} $opt\_E = map(\mathbf{G}_B.E)$ & // 1-to-1 mapping \\ 
    21:& \hspace{4mm} \textbf{for} $e \in opt\_E$ \textbf{do}: & ~ \\
    22:& \hspace{6mm} $\mathbf{G}_B.merge(get\_vertices(e))$ & // Merge cluster vertices \\
    23:& \hspace{2mm} $\mathbf{G}.add(\mathbf{G}_B)$ & // Add $B$'s clusters to $\mathbf{G}$\\ 
    24:& \hspace{0mm} \textbf{for} $c \in \mathbf{G}$ \textbf{do}: & // Iterate final clusters\\
    25:& \hspace{2mm} \textbf{if} $|c| \ge s_m$ \textbf{then}: & // Size at least $s_m$ \\
    26:& \hspace{4mm} $\mathbf{M}.add(c)$ & // Add to $\mathbf{M}$\\ 
    27:& return $\mathbf{M}$ & // Output $\mathbf{M}$ \\ %[1mm]
      \hline %\\ %[2mm]
  \end{tabular*}
\end{table}

\subsection{Late mapping-based clustering} % approach}
\label{subsec:late_mapping}

The early mapping-based approach (described in the previous
section) is efficient in terms of the number of comparisons
required, as we will discuss in Section~\ref{sec-analysis}.
However, since the optimal mapping is conducted between
the records of a database and only the records from the previously
processed databases, 
early mapping can potentially lead to a reduction of the quality of the final linkage results.
In this section, we propose a late mapping-based approach to
improve linkage quality at the cost of more comparisons.  
%%If late one-to-one mapping is used, this step is not required.

In addition to one-to-one mapping
%between vertices 
and merging vertices (as described for the early mapping
approach in the previous section), %Section~\ref{subsec:early_mapping}), 
the late mapping approach involves a 
third phase, which is splitting vertices. %(or clusters).

\smallskip
\textbf{Splitting vertices}: 
%Splitting of a cluster occurs in late mapping,
%as will be described below. %in Section~\ref{subsec:late_mapping}. %,
%if a cluster contains more
%than one record from a single database, because a cluster can only contain
%a maximum of one record from each database. 
%Two different types of splitting are required at different
%stages of late mapping:
%(1) splitting with copy where
%multiple $m$ ($m > 1$) records from the
%same database in a cluster $c$
%are split into $m$ different clusters along with the
%records from other databases that were in $c$,
%and (2) splitting without copy where 
Records $r_{i,x}$ that
belong to a database $\mathbf{D}_i$ in a cluster $c$ are 
split into singletons
containing the records $r_{i,x}$,
while the remaining records from other databases
%(i.e.\ without copying the records from
%other databases 
$r_{j,y} \in \mathbf{D}_j$ are kept in $c$
(with $1 \le i,j \le p$ and $i \neq j$).
%).

We next describe the steps of 
late mapping-based clustering (continuing after the initial steps (1) to (3)).
%With late mapping,
It requires
$p-1$ iterations first to merge records from $p$ parties, 
and then $p$ iterations for splitting and
applying one-to-one mapping,
%An overview of our approach 
%with late one-to-one mapping is 
as illustrated in Figure~\ref{fig:clus_mpprl2}
and 
outlined in Algorithm~2.
%The steps of %our protocol with
%late mapping-based clustering are (continuing after steps (1) to (3)):

\begin{enumerate}[(1)]

\setcounter{enumi}{3}

\item For each block $B$, 
the masked records in $B$ of the first party
$P_i$ ($1 \le i \le p$)
%(unlike early mapping,
(in any order) %is not important at this stage in
%the late mapping approach and it will be required 
%at a later stage as will be described below)
are added into a graph $\mathbf{G}_{B}$ as separate vertices 
(lines~2 to 8 in Algorithm~2). %,
%resulting in $n_i/b$ vertices,
%assuming $|\mathbf{D}_i| = n_i$ and 
%the $block(\cdot)$ function generates
%$b$ blocks of equal size.
Then the second party $P_2$'s  
masked records are inserted
into the graph $\mathbf{G}_{B}$ as separate vertices and the similarities between 
these singleton vertices
of $P_1$ and $P_2$ are calculated (lines~9 to 12). 
If the similarity
between two vertices
is at least $s_t$, then an edge is created between them
(as shown in Figure~\ref{fig:clus_mpprl2} and
described in lines~13 and 14 in Algorithm~2).
%the two vertices are merged into a cluster in $\mathbf{G}_B$
%(as shown in Figure~\ref{fig:clus_mpprl} and
%explained in lines~13-14 in Algorithm~1).

%\item The vertices that are connected by an edge are then merged
%into one cluster (lines~15 and 16), while the vertices that do not have
%any connecting edge (those that are not matching to any vertices
%in the other database) are kept as unclustered vertices.
%In our running example shown in Figure~\ref{fig:clus_mpprl2}, 
%the records $(r_{1,1},r_{2,1})$, $(r_{1,1},r_{2,2})$, 
%$(r_{1,3},r_{2,2})$, $(r_{1,3},r_{2,3})$, and $(r_{1,4},r_{2,1})$
%from parties $P_1$ and $P_2$, respectively,
%are merged into clusters, while $r_{1,2}$ 
%is not clustered
%with any vertices from $P_2$. 
\item This leads to several many-to-many matched vertices 
between the two parties, $P_1$ and $P_2$,
since no early optimal mapping is applied.
The vertices that are connected by an edge are then merged
into one cluster (lines~15 and 16), while the vertices that do not have
any connecting edge (those that are not matching to any vertices
in the other database) are kept as unclustered vertices.

%In this case, a merged cluster $c \in \mathbf{G}_B$ 
%might contain $d$ records ($d > 1$)
%from the second party (i.e.\ $d$ records from the second party have
%a similarity above $s_t$ with the record from the first party).
%Therefore $c$ needs to be split into $d$ clusters where each cluster contains 
%one of the second party's records along with the record from the first 
%party copied from $c$ (using the $split\_with\_copy(\cdot)$ function in lines~17 to 19
%in Algorithm~2). 
%For example, the cluster containing records $(r_{1,1},r_{2,1},r_{2,2})$
%in Figure~\ref{fig:clus_mpprl2}
%is split into two clusters with records $(r_{1,1},r_{2,1})$ and $(r_{1,1},r_{2,2})$,
%respectively.
%The idea of splitting is to have clusters in
%$\mathbf{G}_B$ such that each cluster contains a maximum of one record
%from each party.

\item The $LU$ then proceeds with the masked records of the
remaining parties, where the records are first added in $\mathbf{G}_B$
as singleton
vertices, and then the similarities between these singletons and the
clusters from all previous parties' masked
records are calculated and an edge is created connecting those vertices
that have a similarity of at least the minimum threshold $s_t$ (lines~9 to 14). The vertices connected by an edge are
merged into one (lines~15 and 16),
while the vertices that are not matching to any other vertices
remain as unclustered vertices.

\item This process of merging vertices is repeated %and splitting
until the masked records of all parties are processed 
(i.e. $p-1$ iterations for each block).
The output will be clusters
that are overlapping, 
which means a record from one party might be in several clusters (i.e.\
matching with several sets of records from other parties).
In the example shown in Figure~\ref{fig:clus_mpprl2}, the merged clusters are overlapping. For example, $r_{1,1}$ is in two clusters and $r_{3,4}$ in 3 clusters.
Since %we assume that 
the databases
are deduplicated, a record must be matching only to one set of records from
other databases. Therefore a late one-to-one mapping needs to be applied on
all clusters. 

\item In order to conduct late one-to-one mapping, the parties are 
ordered using the $ord(\cdot)$ function (similar to the early mapping approach),  
and the records of the first party
in the ordered list are split from the clusters into singleton vertices
(one vertex for each unique record) using the $split(\cdot)$ function
in lines~17 to 19.
In the example in Figure~\ref{fig:clus_mpprl2}, records from party $P_1$ are split from the merged clusters into singletons in iteration 1.
The optimal one-to-one mapping is then applied
between the singletons and the clusters containing unique sets of records from
other parties (lines~20 to 22).
The number of edges generated for mapping in each iteration corresponds to the number of clusters that appear before splitting in that iteration. In the running example shown in Figure~\ref{fig:clus_mpprl2}, the number of clusters before iteration 1 is $8$ (initial merged clusters) and therefore iteration 1 generates $8$ edges between $P_1$'s singletons and other parties' clusters.
This results in the first party's records being clustered with the
highly matching set of records from other parties.
For example, $r_{1,1}$ is mapped and merged with the cluster containing $(r_{2,2},r_{3,1})$, while $r_{1,2}$ is merged with the cluster of records $(r_{2,1},r_{3,2})$.

The process is
repeated for all parties ($p$ iterations) in the ordered list until the 
set of non-overlapping
clusters is obtained. %, as shown in Figure~\ref{fig:clus_mpprl2}. 
As shown in Figure~\ref{fig:clus_mpprl2}, the set of non-overlapping clusters are generated after 3 iterations of splitting and merging of clusters (with one party's records at an iteration) for linking $p=3$ databases. 

It is important to note that late mapping
requires more cluster comparisons than early mapping, 
as it does not prune edges at an early stage potentially 
leading to many merged clusters. However, it potentially results in better 
linkage quality
since the late one-to-one mapping considers all parties' records,
unlike in early one-to-one mapping where only the previous parties' records are
considered.  

\begin{figure*}[!th]
\centering
\scalebox{0.91}[0.8]{\includegraphics[width=0.89\textwidth]{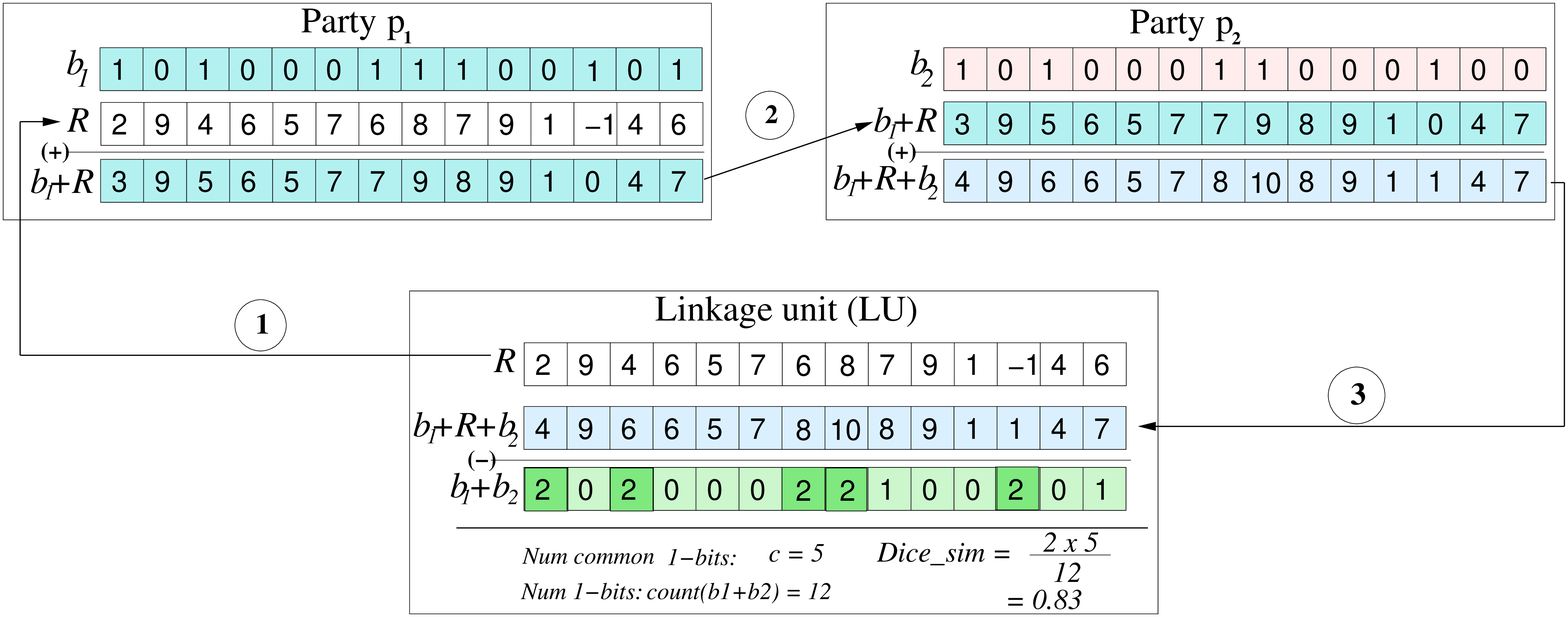}}
\caption{\small{An example of using CBFs generated from two BFs $b_1$ and $b_2$ from two respective parties $P_1$ and $P_2$ using a secure summation protocol for similarity calculation, as described in Section~\ref{sec-privacy_improvement}.
}
}
\label{fig:cbf}
\end{figure*}

\item The final (non-overlapping)
clusters of block $B$ (i.e.\ vertices in graph $\mathbf{G}_B$)
are added to $\mathbf{G}$ (line~23).
%These final (non-overlapping) clusters in graph $\mathbf{G}$ 
%represent the sets of matching records either across all $p$ parties
%or across any subsets of parties.
The final clusters $c \in \mathbf{G}$ with $|c| \ge s_m$ 
(i.e.\ vertices containing matching records from at least $s_m$ parties) 
are added to the final matching set of records $\mathbf{M}$ (lines~24 to 27). 
%The clusters of size less than $s_m$ (with $s_m \le p$) are pruned to identify
%matching records from at least $s_m$ parties.

\end{enumerate}

\subsection{Improving Privacy}
\label{sec-privacy_improvement}

Our clustering-based MP-PPRL protocol can be used with any encoding/masking technique, such as BF encoding~\cite{Sch09} as used in the example described in Figure~\ref{fig:bloomfilter}. BF encoding is one of the widely used methods in PPRL due to its efficiency compared to cryptographic methods and controllable/tunable privacy-accuracy trade-off~\cite{Vat13,Sch09,Bro19}.

However, BFs are susceptible to inference attacks by adversaries as has been shown in several studies~\cite{Chr18,Chr18b,Kuz11,Nie14}. %Kro15, 
Counting Bloom filter (CBF), a variation of BF (as described in Section~\ref{sec-preli}), provides improved privacy guarantees compared to BF for multi-party PPRL~\cite{Vat16b}. %CBF is an integer vector that stores counts of elements mapped to it, rather than bits of 1s and 0s represented by presence and absence of elements mapped to it in a BF. Multiple BFs can be summarized into a single CBF by applying a vector addition operation on all BFs. The similarity of multiple BFs can be calculated using its respective CBF. As proven in~\cite{Vat16b}, such CBFs improve privacy as only the summary information (counts) is revealed, not the individual bits.
We therefore adapt the CBF-based approach for our protocol to improve privacy against inference attacks. Instead of all parties sending their records' BFs to a $LU$, they can generate CBFs from the BFs using a secure summation protocol, as shown in Figure~\ref{fig:cbf}. For example, in the first iteration the first two parties participate in a secure summation protocol~\cite{Clif02} with a $LU$ and generate CBFs for every pair of BFs. 

In the basic secure summation protocol~\cite{Clif02}, the $LU$ provides a random vector $R$ to the first party, which adds its BF $b_1$ to $R$ and sends the summed vector $b_1+R$ to the second party. The second party then adds its BF $b_2$ to the received sum and sends back the final summed vector $b_1+b_2+R$ to the $LU$. The $LU$ subtracts the random vector $R$ from the received sum to generate the CBF $c = b_1+b_2$. Using the generated CBFs, the $LU$ calculates the similarities of pairs of BFs from the two parties (Equation~\ref{eq:Dice_coefficient_cbf}). 

In the second iteration the $LU$ already has the possible matches (clusters) from the first two parties. A secure summation protocol is then used by the first three parties and the $LU$ to generate CBFs from all matches identified in the first iteration along with every BF from the third party. Note that every iteration requires different BF encoding by the corresponding parties to avoid the $LU$ learning the new party's BFs. This is repeated until all parties' records are compared by the $LU$.

As discussed in Section~\ref{subsec_privacy_analsis} in detail, CBFs are less vulnerable to inference attacks. However, they incur memory cost ($l \times \lceil log_2(p) \rceil$ for $p$ BFs of $l$ bits) as well as communication costs. Every $i^{th}$ iteration requires $i+1$ communication steps in the secure summation protocol (for example, when $p=2$, the number of communication steps required for secure summation is $3$, as illustrated in Figure~\ref{fig:cbf}) to generate the CBFs. %The communication and computation complexities are discussed in detail in the following section.

% --------------------------------------------------------------------

\section{Analysis of the Protocol}
\label{sec-analysis}

In this section we analyze our MP-PPRL protocol
with regard to complexity, privacy, and linkage quality.

\subsection{Complexity Analysis}
\label{subsec_comp_analsis}
Assume $p$ parties participate in the linkage
of their respective databases, each containing
$n$ records, and $b$ blocks are generated by the
blocking function, each block containing $n/b$ records. 
%It is important to note that
%the ordering of databases for incremental clustering
%according to their sizes in ascending order
%can improve the efficiency as smaller number of
%clusters will be generated in early iterations 
%reducing the total number of comparisons.
%
In step~(1) of our protocol 
(as described in Section~\ref{sec-protocol}),
masking records (with $g$ average $q$-grams per record) into
BFs of length $l$ using $k$ hash functions for $n$
records is $O(n\cdot g\cdot k)$ for each party.
Blocking the databases in step~(2) has $O(n)$
computation 
and $O(p\cdot b)$ communication complexity 
(assuming $b \le n$ blocks)
at each party. 
%The communication complexity of both early and late mapping-based
%approaches is linear in $p$ and $n$, as
In step~(3),  
$n$ masked records from $p$ parties need to be %masked and
sent to the $LU$ for conducting the linkage, which is
of $O(p \cdot n)$ communication complexity.

%Our graph-based 
The early mapping-based approach for incremental clustering
%approach with early one-to-one mapping 
has guaranteed quadratic %and linear
worst case computation
complexity
in both $p$ and $n$.
The worst case (in terms of the number of comparisons required)
occurs with early one-to-one mapping
in two ways:
when merging vertices from a database $\mathbf{D}_i$
with the vertices in the graph $\mathbf{G}$,  
(a) no vertices (records) in $\mathbf{D}_i$ 
match with vertices in $\mathbf{G}$ resulting in $n$ additional
singleton vertices in every iteration, or 
(b) every vertex/record in $\mathbf{D}_i$ matches with
a vertex in $\mathbf{G}$ resulting in $n$ 
vertices with one additional record in every iteration
leading to $n$ final vertices %(or clusters)
containing $p$ records.
%
%\smallskip
%\begin{proposition}
%\label{prop:cbf}
%In the worst case, the computation complexity of our graph-based incremental clustering approach with early one-to-one mapping is $O(n^2/b \cdot p(p-1))$ (quadratic complexity with $p$ and $n$).
%%, while its communication complexity is $O(p \cdot n)$ (linear complexity with $p$ and $n$).
%\end{proposition}
%
%\smallskip
%\begin{proof}
The protocol requires $p-1$ iterations for mapping and merging records from $p$ databases in each of the $b$ blocks.
%Linking the first two databases requires $O(n^2/b)$ comparisons which
%results in the worst case of either $2n$ singletons (if none of the records
%are matching and therefore not merged)
%or $n$ merged clusters each containing $2$ records (one from each database). 
%Next either the $2n$ singletons
%need to be compared with $n$ records from the third database,
%or each of the $n$ records from the third database
%needs to be compared with 
%all clusters each contains $2$ records,
%requiring $O(2n^2/b)$ comparisons in both cases.
Generally, comparing $n$ records from $\mathbf{D}^M_i$
with vertices in $\mathbf{G}$ 
in the worst case is of $O((i-1)n^2/b)$, with $1 \le i \le p$.
Therefore, the total worst case complexity is 
$O(n^2/b + 2n^2/b + 3n^2/b + \cdots + (p-1)n^2/b)$,  
%$= O(n^2/b \cdot p/2 \cdot (p-1))$, 
    which is $O(n^2/b \cdot p^2)$. %$O(n^2/b \cdot p(p-1))$.
%\end{proof}

%\smallskip
The late mapping-based approach
has an exponential computation complexity
in the worst case scenario assuming each record 
from a database is matched to all records in
all other databases (due to the many-to-many matching),
leading to 
$O(b \cdot (n/b)^p)$ overlapping final clusters each containing $p$ records. 
However, assuming the databases are individually deduplicated (as discussed in Section~\ref{sec-preli})
and an appropriate similarity threshold is used for merging clusters, 
only a small number of additional clusters ($o$ with $o \ll n^2$) are 
generally generated in each iteration ($n + o$ merged clusters in total).
This leads to an average computation complexity of 
$O(n(n)/b + n(n+o)/b + 2n(n+o)/b + \cdots + (p-1)n(n+o)/b) = p(p-1) \cdot n \cdot (n+o)/b $,
%
% Peter, 7 February: *********************************
%which is $O(p(p-1) \cdot (n^2+no)/b )$.
which is $O(p^2 \cdot (n^2+no)/b )$.
%
%The number of comparisons required when 
%iteratively splitting records of each database
%from these merged clusters and applying optimal (late) one-to-one mapping
%with clusters containing all other records
%is $O(p \cdot n \cdot p(n + o))$.
Therefore, the computation complexity of late mapping in the
average case is quadratic in both $p$ and $n$.

Overall, our MP-PPRL protocol
%for linking multiple databases
% Peter, 7 February: *********************************
has a worst-case quadratic 
% has a quadratic
computation complexity 
and a linear communication complexity 
in the number of records $n$ and databases $p$, 
%in both $n$ and $p$,
which are both significantly lower than 
the exponential complexities of earlier
MP-PPRL protocols~\cite{Vat14c,Vat16b,Lai06}. 
Please note that extending existing 
PPRL techniques (that can link two databases with quadratic complexity) to multi-database linkage requires the additional step of clustering once the pair-wise similarities have been calculated.
% Peter, 7 February: *********************************
Investigating other clustering algorithms that have been developed for record linkage~\cite{Has09c,Nan19,Sae18} in the context of MP-PPRL is subject to future research. 
%. Our protocol thus significantly
%reduces the exponential complexity
%of earlier MP-PPRL protocols~\cite{Lai06,Vat16b}
%to quadratic complexity.
%In Section~\ref{sec-experiment} we empirically evaluate  
%the scalability of our proposed approach
%for MP-PPRL compared to existing MP-PPRL
%approaches that have exponential complexity.

\subsection{Privacy Analysis}
\label{subsec_privacy_analsis}

As with most existing PPRL approaches,
we assume that all parties follow
the 
honest-but-curious 
adversary model~\cite{Lin09}, %Vat13}, 
where
%In this model 
the parties follow the protocol while being 
curious to find out as much as possible 
about 
%learn 
the other parties' data 
by means of inference attacks on masked records or by colluding with other parties~\cite{Vat14}.
%To analyze the privacy against inference attacks, 
%we discuss what the parties can learn from the data exchanged  
%during the protocol.
%We consider the blocking step as a black box where any 
%existing private blocking technique~\cite{Chr12,Ran14,Ran16} can be used.
We assume the private blocking technique used 
(as a black box) does not reveal any
sensitive information to any parties,
and the blocks generated meet the required 
privacy guarantees, such that each block contains at least a minimum number ($k$) of records~\cite{Vat14} or are differentially private~\cite{Kuz13}, to overcome
frequency attacks. %If each block generated contains at least $k$ records, then frequency attack using the frequency information of blocks becomes difficult~\cite{Vat14}.

%Communication occurs during 
In the matching step
%only once when 
the parties send their 
masked records (BFs) to the $LU$ to conduct the
linkage. 
In order to overcome inference attacks by the $LU$
on the BFs, the counting Bloom filter (CBF)-based approach 
(described in Section~\ref{sec-privacy_improvement})
can be applied
where the $LU$ sequentially gets a CBF from the relevant 
set of parties for each cluster of records. Using the CBF
the $LU$
can calculate cluster similarity as equivalent 
to calculating cluster similarity
using individual BFs~\cite{Vat16b}. 
%For example, in iteration 1 in our protocol, instead of the parties
%sending their BFs to the $LU$ they generate a CBF for each pair of BFs
%by applying a secure summation
%protocol on their BFs, and then send the resulting CBFs (summed vector of BFs)
%to the $LU$. 
CBFs significantly reduce the risk of inference
attack compared to BFs~\cite{Vat16b}. %, as 
%theoretically 
%proven in~\cite{Vat16b}.
An inference attack allows an adversary to map a list of known values from a global dataset (e.g.\ $q$-grams or attribute values from a public telephone directory) 
to the encoded values (BFs or CBF) using background information (such as frequency)~\cite{Kuz11,Vat14}.
The only information that can be learned from such an inference attack using a CBF $c$ of a set of $x$ BFs 
(summed over $x$ parties) is if a bit position in $c$ is either $0$ or $x$ which means it is set to $0$ or $1$, respectively, in the BFs of all $x$ parties.

\begin{prop}
The probability of identifying the unencoded (original)
values of $x$ ($x > 1$) individual records $R_i$ (with $1 \le i \le x$) given a single CBF $c$ is smaller than the probability of identifying the unencoded values of $R_i$ given $x$ individual BFs $b_i$, $1 \le i \le x$.
\begin{eqnarray}
\label{eq:pr_inf} 
\forall_{i=1}^{x} Pr(R_i|c) < Pr(R_i|b_i) \nonumber
\end{eqnarray}
\end{prop}

\textbf{Proof:}
Assume the number of original (unencoded) values that can be mapped to a masked BF pattern $b_i$ from an inference attack is $n_g$. $n_g = 1$ in the worst case, where a one-to-one mapping exists between the masked BF $b_i$ and the original unencoded value of $R_i$. The probability of identifying the original value given a BF in the worst case scenario is therefore $Pr(R_i|b_i) = 1/n_g = 1.0$~\cite{Vat14}. However, a CBF represents $x$ BFs and thus at least (in the worst case) $x$ original (unencoded) values, which leads to a maximum of $Pr(R_i|c) = 1/x$ with $x>1$ (when $x=1$, then $c \equiv b_i$). Hence, $\forall_{i=1}^{x} Pr(R_i|c) < Pr(R_i|b_i)$.
%\end{proof}

Further, the collusion-resistant secure summation protocols described in~\cite{Vat16b,Tas13}
can be used to overcome the risk of collusion among the parties in order to learn about another party's data.
%
%We assume a trusted
%$LU$ is available, which does not collude with any parties,
%as is commonly used in practical PPRL applications~\cite{Ran13}.
%%Collusion is still a possible privacy risk for both $LU$-based and non-$LU$-based MP-PPRL approaches.
%Inference attacks on BFs (e.g.\ cryptanalysis attack)
%are possible as empirically shown in several studies~\cite{Sch15,Nie14}.
%
%Several hardening techniques for BFs have been used to
%overcome such inference attacks. For example, using record-specific data,  
%such as the date or year of birth as the salting key for hashing,
%using random hashing, or 
%rehashing of already hashed q-grams of QIDs are some of the techniques
%developed for hardening BFs against inference attacks~\cite{Sch15}. 
%In our experiments, 
We also use the cryptographic long term key (CLK) encoding~\cite{Sch15}
as a BF hardening method,
where QID values of a record are hash-mapped into a record-level BF. %which 
This approach improves privacy against inference attacks by decreasing the probability of
suspicion~\cite{Vat14}.
%
%Therefore, assuming secure blocking and appropriate hardening methods
%for BF-based data masking, our protocol is
%secure under the HBC model.

\subsection{Linkage Quality Analysis}
\label{subsec_quality_analsis}

Our MP-PPRL protocol allows approximate matching
of QID values, in that
data errors and variations are taken into account depending upon
the minimum similarity threshold $s_t$ used. 
Further, our protocol allows subset matching by identifying matching
records across any subset of databases.
This improves the linkage quality of MP-PPRL
where records of a single entity can be either in all 
databases or in a subset of databases only (which is often a realistic scenario
in practical applications). 
To the best of our knowledge, this is the first approach that addresses
subset matching for MP-PPRL.

The two proposed methods of early and late one-to-one mapping
in the incremental clustering approach have a trade-off between complexity
and linkage accuracy. As analyzed in Section~\ref{subsec_comp_analsis}, the early mapping approach has lower computational complexity than the late mapping approach. In the following, we analyze the linkage quality of these two mappings. 

Conducting early one-to-one mapping in every iteration 
before merging clusters
significantly reduces the computation complexity (as discussed in
Section~\ref{subsec_comp_analsis}). 
However, this approach might reduce linkage quality, because when conducting
optimal one-to-one mapping with $P_i$'s records then only the records from the previous parties ($P_1$ to $P_{i-1}$, with $1 < i \le p$)
are considered.
In contrast, the late one-to-one mapping is conducted 
for each
party $P_i$'s records considering
records from all other parties $P_j$,
with $1 \le i,j \le p$ and $i \neq j$. Therefore, late mapping
can improve linkage quality at the cost of more comparisons.

In addition, the linkage quality of our protocol depends on the blocking and the deduplication techniques applied on each database. The higher the quality of deduplication results the 
better the one-to-one mapping achieved in our approach 
will be, leading
to higher linkage quality. 
The parties can also be ordered using different ordering functions considering
the known quality of their databases or the
quality of deduplication results, such that the
best quality database is processed first, as
the initial clusters will then be of higher quality leading to
higher quality clustering in the later iterations~\cite{Nen18}.
%to improve the quality of the incremental clustering approach.

Similarly, the average similarity function we use adds a masked record to a cluster if its similarity on average is high with all masked records in the cluster. Different similarity functions, such as minimum similarity (known as complete linkage), where a masked record needs to have a high similarity with all masked records in the cluster, or maximum similarity (single linkage), where a masked record needs to have a high similarity with at least one masked record in the cluster, would have different impacts on the linkage quality. We leave investigating the impact of different ordering and similarity
functions on the linkage quality and efficiency as a future work.

% --------------------------------------------------------------------

\section{Experimental Evaluation}
\label{sec-experiment}

%In this section, 
We empirically evaluate 
the performance of
our MP-PPRL protocol (named as \textbf{AM-Clus})
with the two proposed variations,
early mapping (\textbf{EMap}) and late mapping (\textbf{LMap}), as well as the baseline greedy mapping (\textbf{GMap}),
in terms of scalability, %(complexity), 
linkage quality, and privacy.
In the following sub-section we first describe the
datasets we use in our evaluation. In
Section~\ref{sec-competing-methods} we discuss the
baseline methods to which we compare our proposed
clustering approaches, and in
Section~\ref{sec-measures} the evaluation measures we
employ in our experiments. In Section~\ref{sec-settings}
we then describe our experimental setting, and in
Section~\ref{sec-discussion} we provide an extensive
discussion of the results we obtained in these
experiments.

\subsection{Datasets}
\label{sec-datasets}
One problem with regard to datasets for evaluating MP-PPRL approaches is that there are no datasets available that are generated from multiple parties. Therefore, the general approach to conduct experiments is using multiple datasets sampled with overlap from a single large dataset.
We conducted our experiments on three collections of datasets (including a health dataset): 
%\begin{itemize}
%
%\item

(1) \textbf{NCVR}: A set of datasets generated  
based on the North Carolina Voter Registration
(NCVR) database (available from \texttt{https://dl.ncsbe.gov/}).
We extracted 5,000 to 1,000,000 records
for 3, 5, 7, and 10 parties with 25\% of matching
records across all parties and 25\% of matching
records across subsets of parties. Note that these datasets are different than the NCVR datasets used for the experimental evaluation conducted in~\cite{Vat14c,Vat16b} for the MP-PPRL approaches that allow full-set matching only (not subset matching). The difference is that these datasets contain 25\% of matching records across any subset (of different sizes) of parties and 25\% of matching records across all parties, whereas in the datasets used by~\cite{Vat14c,Vat16b} 50\% of matching records appear across all parties.

We also sampled 10 datasets each containing 10,000
records such that 50\% of records are non-matches and 5\% of records are true matches across each different subset size of $1$ to $10$ ($1,2,3,\cdots,9,10$),
i.e. 45\% of records 
are matching in any 2 datasets while only 5\% of records
are matching in any 9 out of all 10 datasets.
Ground truth
is available based on the voter registration identifiers to
allow linkage quality evaluation.
%We used first name, last name, city, and zipcode attributes
%as the QIDs for the linkage.
%

We generated another series of 
datasets for each of the
datasets generated above, 
where we included 20\% and 40\% synthetically
corrupted records into the sets of overlapping/matching records 
(labelled as 'Corr-20' and 'Corr-40' in the plots, respectively) 
%with levels of 20\% and 40\% 
using the GeCo tool~\cite{Tra13}.
For example, if $p=10$ datasets containing 10,000 records each are linked where 5,000 records are matching across at least 2 of these datasets (minimum subset size is 2 in this example), then 1,000 and 2,000 records from these set of true matches are corrupted, respectively.
%For example, a set of $P=10$ datasets each containing 10,000 records are linked where the number of true matches is 5,000 records, then 1,000 and 2,000 records in these matches are corrupted for 20\% and 40\% corruption levels.
We applied various corruption functions from the GeCo tool on randomly selected attribute values, including character edit operations
(insertions, deletions, substitutions, and transpositions), and
optical character recognition and phonetic modifications
based on look-up tables and corruption rules~\cite{Tra13}. Since the matching records (either one or many in the set chosen randomly) are corrupted, the linkage quality will drop which allows us to evaluate how real data errors impact the linkage quality.

(2) \textbf{NCVRT}: We have
downloaded the NCVR database every second month since October 2011 and
built a combined temporal database of 
26 such datasets 
(i.e.\ 26 snapshots)
each containing over 5 million records
of voters~\cite{Chr13NC}.
Voter registration identifiers are unique which provides
ground truth for evaluation.
%The first name, last name, city, and zipcode were used
%as the QID attributes for the linkage.
This real temporal
dataset allows conducting large-scale experiments
for MP-PPRL assuming each snapshot corresponds to
the dataset from one party ($p=26$ parties).

%\item
(3) \textbf{NSWE}: The third dataset is a real
New South Wales (NSW) emergency presentations dataset from Australia. A previous study that evaluated our proposed method on this sensitive data by the Centre for Data Linkage at  Curtin University provided the presented results~\cite{Ran16b}. In this study, subsets of NSWE dataset were extracted for 5 parties
each containing more than 700,000 records with no duplicates. 
%The first name, surname,
%date of birth, sex, address, and postcode were
%used as QIDs for
%linking records.
These datasets were linked
by the Centre for Health Record Linkage in Sydney providing ground
truth for the linkage~\cite{Ran13}.

%\end{itemize}

\subsection{Baseline Methods}
\label{sec-competing-methods}

As reviewed in Section~\ref{sec-related}, only two %state-of-the-art 
MP-PPRL techniques are available for approximate matching of
string data using probabilistic data structures~\cite{Vat14c,Vat16b}. We use these two
as the baseline approaches to compare our proposed 
approach as they are closely related to our work. %We also use the exact matching MP-PPRL protocol
%proposed by Lai et al.~\cite{Lai06} to evaluate how
%our protocol for approximate matching outperforms 
%such an exact matching protocol
%in the context of dirty data.
We name these approaches as \textbf{AM-BF}
and \textbf{AM-CBF} for the approximate matching
approaches based on BF~\cite{Vat14c} and CBF~\cite{Vat16b}, respectively. 
%and \textbf{EM-BF} for the exact matching approach
%based on BF~\cite{Lai06}.

\subsection{Evaluation Measures}
\label{sec-measures}

We evaluate
the complexity (scalability) of linkage using \emph{runtime}  
and \emph{memory size} required for the linkage. 
The quality of the achieved linkage is measured using the 
\emph{precision}, \emph{recall}, and
\emph{F-measure}, calculated on classified matches and non-matches,
that have widely been used in record linkage, 
information retrieval and data mining~\cite{Chr12}. 
%More 
The ground truth is available for all datasets with known labels of true matches/clusters (either from all $p$ databases for full set matching or subset ($<p$) of databases for subset matching). For example, the \textbf{NCVR}-10000 datasets for $p=10$ parties contain $25\%$, i.e.\ 2,500, record sets (clusters) as true matches from all $p=10$ parties and $25\%$, i.e.\ 2,500, record sets as true matches from any subset ($< 10$) of parties.

Based on the classification of the number of true matching record pairs, TM, from each resulting clusters (either from all $p$ or subsets of databases), false matches, FM, %true non-matches, TN,
and false non-matches, FN, the linkage quality measures are defined as below~\cite{Chr12}:
\begin{enumerate}
\item Precision: the fraction of record pairs in all clusters classified as matches by the PPRL classifier that are true matches: $TM/(TM+FM)$

\item Recall: the fraction of true matches in clusters that are correctly classified as matches by the classifier: $TM/(TM+FN)$ 

\item F-measure: harmonic mean of Precision and Recall: $2 \times (Precision \times Recall)/(Precision+Recall)$
\end{enumerate}

%Peter added 2014020
We use the F-measure in our evaluation to allow comparison
with related earlier publications. We however note
that recent research~\cite{Han18} has identified some
problematic issues when the F-measure is used to compare
record linkage classifiers at different similarity
thresholds. This work is ongoing and there is currently
no accepted appropriate new measure that combines precision
and recall into one single meaningful value.

In line with other work in PPRL~\cite{Vat16b,Ran14,Sch15}, we
evaluate privacy using disclosure risk (DR) measures based on the
probability of suspicion $P_s$, 
i.e.\ the likelihood a masked (encoded) database
record in $\mathbf{D}^M$ can be matched with one or several
known values in a
publicly available global database $\mathbf{D}_G$. 
The probability of suspicion for a masked record $r^M$, $P_s(r^M)$,
is calculated as $1/n_g$ where
$n_g$ is the number of possible matches in $\mathbf{D}_G^M$
to the masked record $r^M$.
We conducted a linkage attack~\cite{Vat14} 
%by mapping the bit patterns in the BFs
%generated from $\mathbf{D}^M$
%to the BFs generated from $\mathbf{G}^M$
assuming the worst case scenario of $\mathbf{D}_G \equiv \mathbf{D}$,
and the BF hash functions are known to the adversary.
Based on such a linkage attack, 
we calculate 
\begin{enumerate}
\item mean disclosure risk ($DR_{Mean}$):
the average risk 
($\sum_i^{|\mathbf{D}^M|}$ $P_s(r_i^M)/|\mathbf{D}^M|$) 
of any sensitive value in $\mathbf{D}^M$ being re-identified~\cite{Vat14}
\item marketer disclosure risk ($DR_{Mark}$): the proportion of masked records in $\mathbf{D}^M$ that
match to exactly one masked record in $\mathbf{D}_G$ ($|\{r_i^M \in \mathbf{D}^M: P_s(r_i^M) = 1.0\}| /|\mathbf{D}^M|$)
\end{enumerate}

\begin{figure*}[ht!]
\centering
 \includegraphics[width=0.42\textwidth]{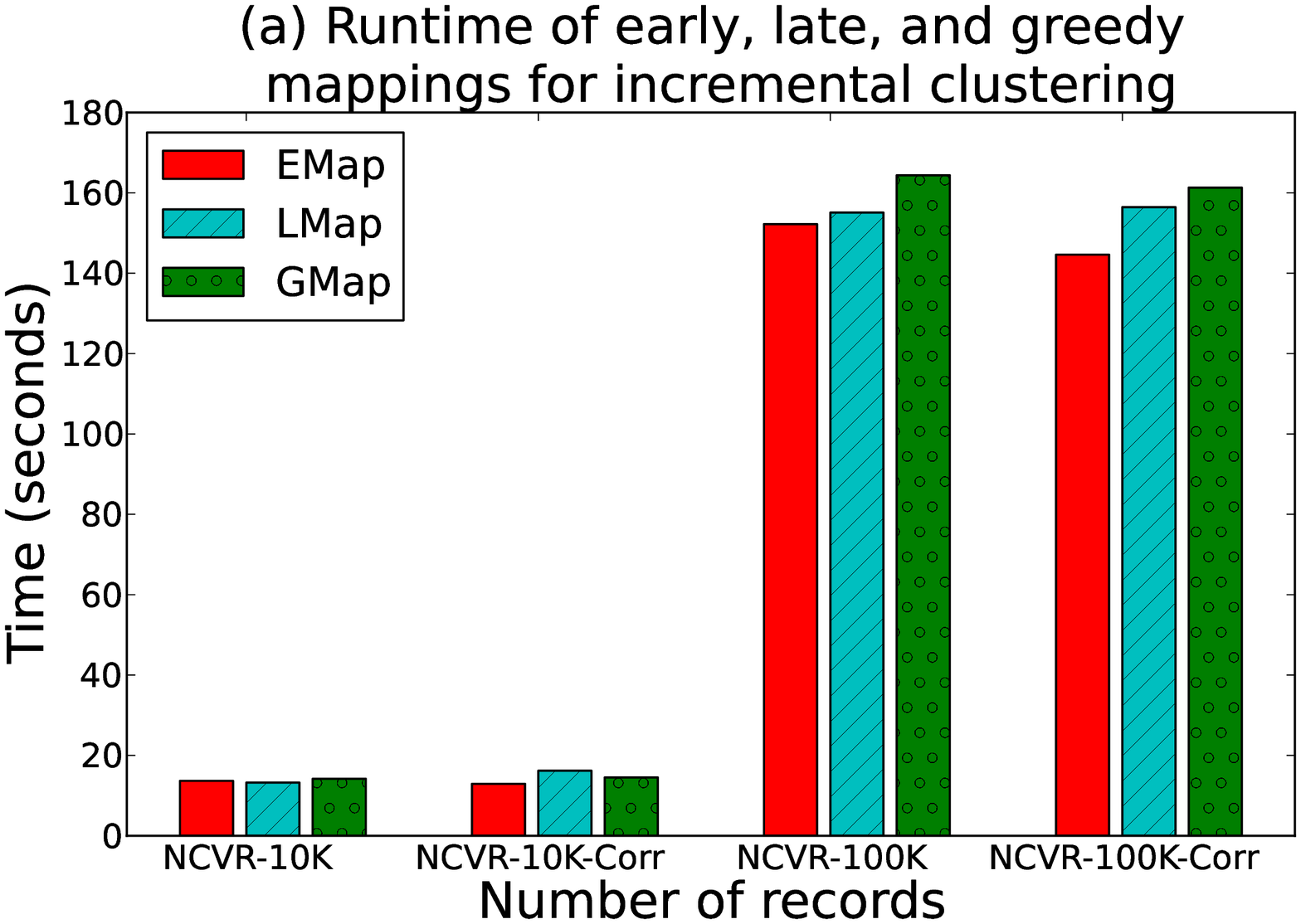}
~~~~~~~~~~
 \includegraphics[width=0.42\textwidth]{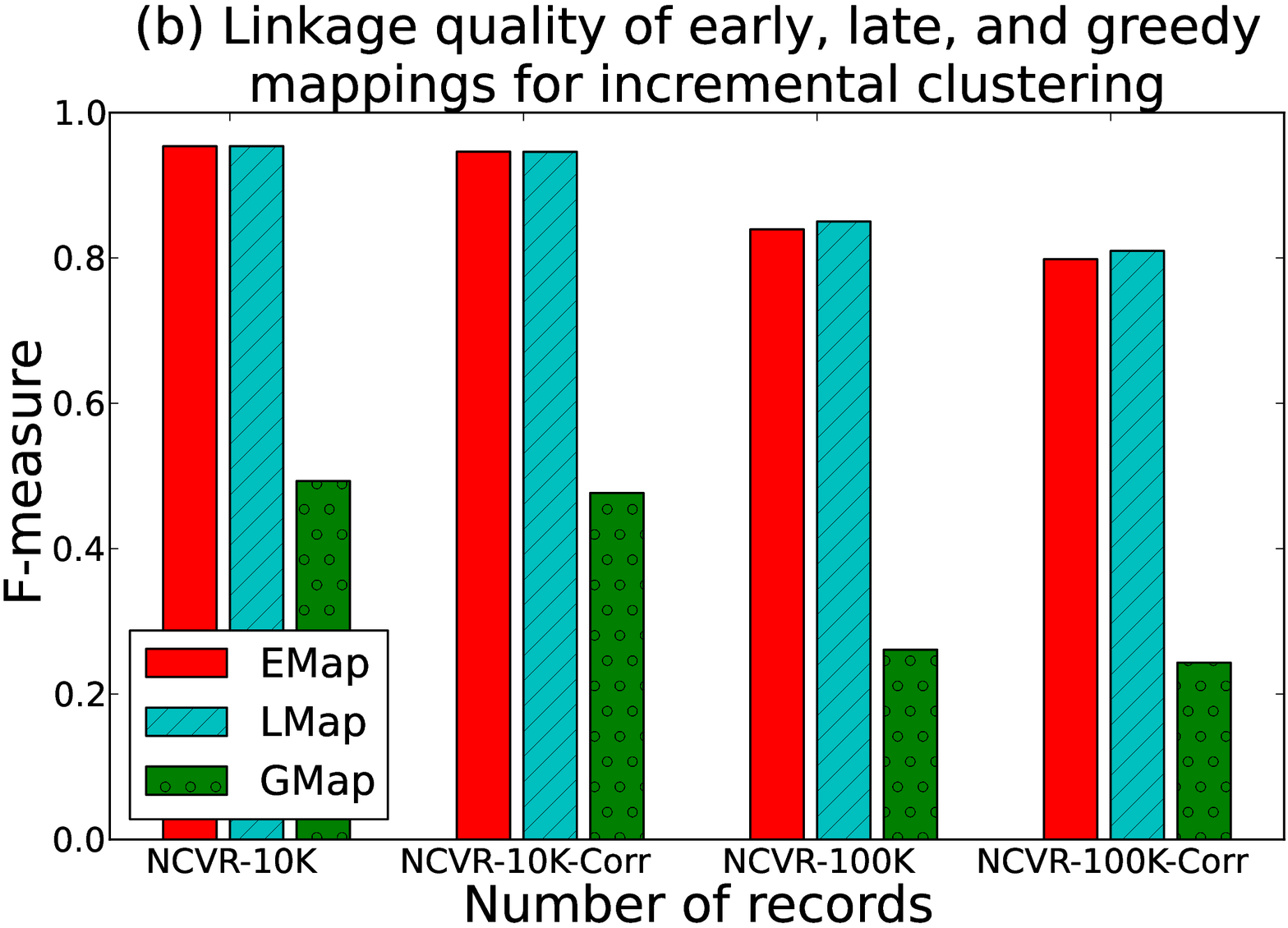}
  \caption{\small{Comparison of (a) runtime and (b) F-measure 
                  for the early mapping (EMap), late mapping (LMap), and baseline greedy mapping (GMap) approaches of incremental clustering on \textbf{NCVR} datasets.}
    }
\label{fig:map}
\end{figure*}

\subsection{Experimental Setting}
\label{sec-settings}

%Record-level BF encoding (CLK)~\cite{Sch11}
%is used for
%masking database records.
Following earlier BF work in 
PPRL~\cite{Dur13,Vat14c,Sch15}, we set
the parameters as BF length $l=1,000$,
the number of hash functions $k=30$, and the
length of grams (substrings of QIDs) is $q=2$.
Soundex-based phonetic encoding~\cite{Chr12} 
is used as the blocking function. %Last name only for all the experiments on \textbf{NCVR} datasets except  and combination of first and last names for most of the experiments on \textbf{NCVR} datasets except  
%and \textbf{NCVRT} datasets, respectively, are used as the blocking key attributes.
%
The last name is used as the blocking key for the scalability experiments for the different sizes of \textbf{NCVR} datasets, while a combination of first and last name attributes are used as the blocking key attributes for other experiments on the \textbf{NCVR} and \textbf{NCVRT} datasets 
%and last name only for 
%most of the \textbf{NCVR} experiments except for the comparative experiments with baseline approaches which used the combined first and last names as blocking keys 
due to the large runtime requirement. Using last name as the blocking key results in larger blocks and thus requires longer runtime. However, larger blocks improve privacy against frequency attacks on blocks, which is preferred in privacy-preserving applications~\cite{Vat13}.

We also used an existing
multi-party private blocking function using bit-trees~\cite{Ran14} on the surname and date of birth attributes
for linking the \textbf{NSWE} datasets.
Soundex-based phonetic blocking on first and last names provides an average pairs completeness (similar to recall it calculates the percentage of true matches found in the candidate record sets generated by a blocking method~\cite{Chr12}) of $0.98$ and the bit-trees-based blocking on surname and date of birth values provides $0.96$ pairs completeness.
We used the first name, last name, city, and zipcode attributes
as QIDs for the linkage of the \textbf{NCVR} and \textbf{NCVRT} datasets,
while first name, surname,
date of birth, sex, address, and postcode are
used as QIDs for linking records in the \textbf{NSWE} datasets. These attributes are commonly used personal identifying attributes for linking records across databases~\cite{Chr12}.

\begin{figure*}[ht!]
\centering
 \includegraphics[width=0.42\textwidth]{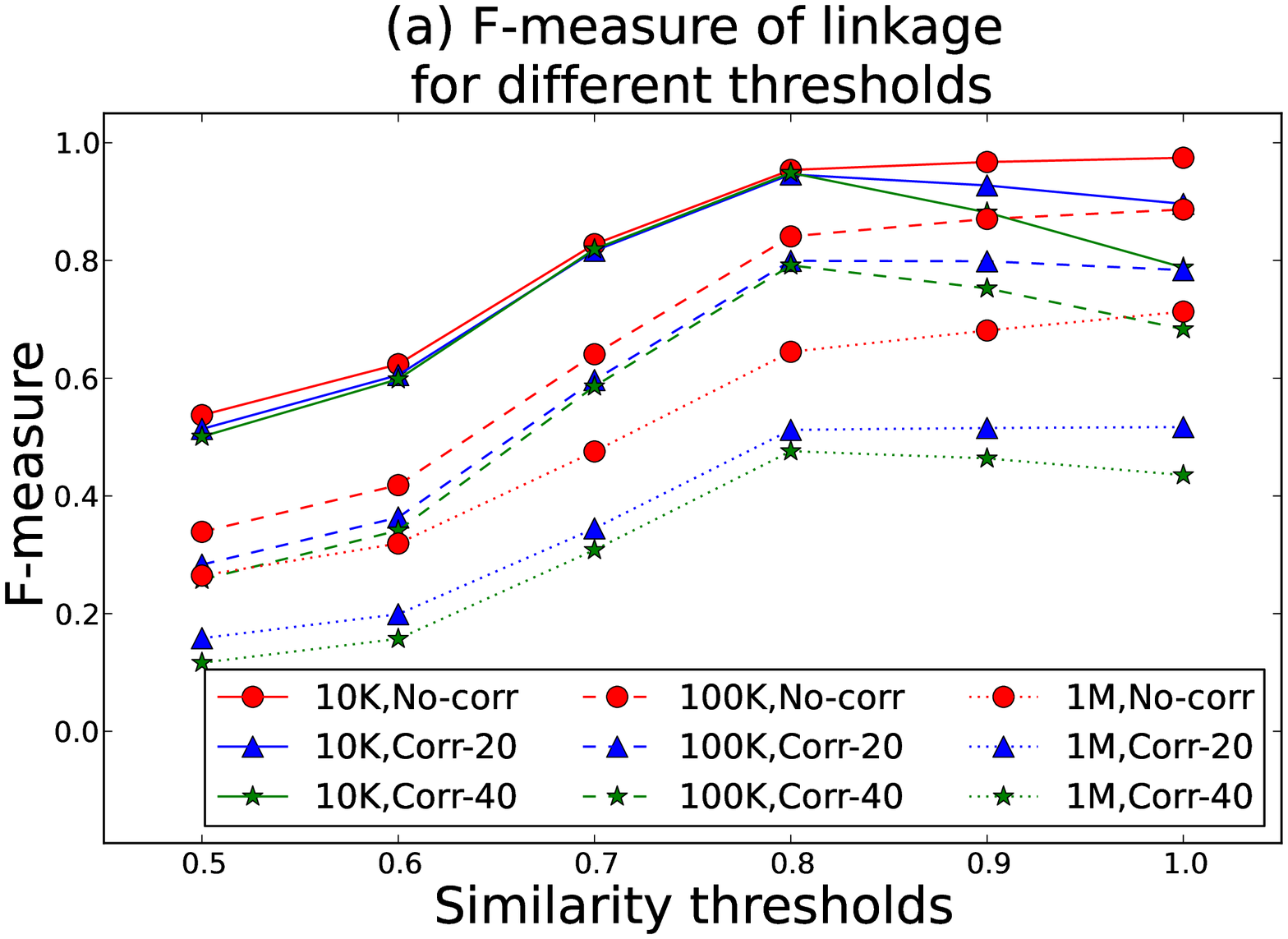}
~~~~~~~~~~
 \includegraphics[width=0.42\textwidth]{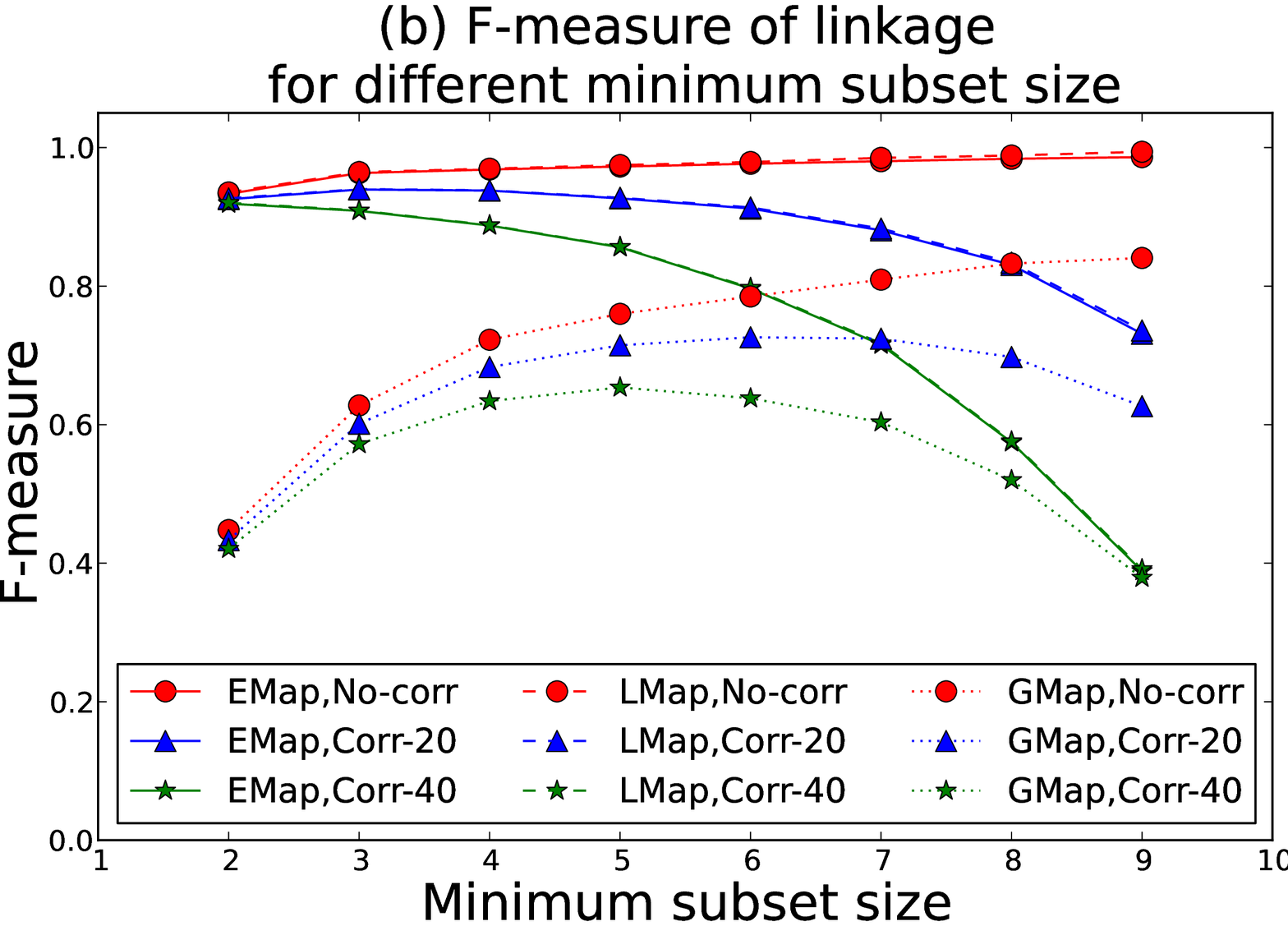}
 
  \caption{\small{F-measure of linkage for (a) different similarity thresholds on \textbf{NCVR}-10K, 100K, and 1M datasets and (b) different minimum subset size on \textbf{NCVR}-10K subset datasets (as described earlier) for $p=10$.}
    }
\label{fig:param}
\end{figure*}

We implemented both our proposed approaches and the competing 
baseline approaches in Python 2.7.3,
and ran all experiments on a server with 
four 6-core 64-bit Intel Xeon
2.4 GHz CPUs, 128 GBytes of memory and running Ubuntu 14.04. The
programs and test datasets
(except \textbf{NSWE}) are available from the authors.

%
%We did not achieve significant differences between different
%similarity functions (minimum similarity, average similarity)
%as described in Section~\ref{sec-protocol}, 
%and therefore we used the average
%similarity for all the other experiments.
%We used random ordering for the ordering function
%(since the datasets are of similar quality
%and have similar sizes,
%which therefore do not 
%provide significant differences
%when ordering by quality or size, respectively).
%In future, we aim to conduct more experimental studies
%with different datasets.

\subsection{Discussion}
\label{sec-discussion}

In this section we discuss the results of our experimental study.

%\begin{figure*}[ht!]
%\centering
% \includegraphics[width=0.32\textwidth]{computational_complexity_p_comp}
%~
% \includegraphics[width=0.32\textwidth]{Accuracy_p}
%~
% \includegraphics[width=0.32\textwidth]{Privacy_p_mean}
% 
%  \caption{\small{Comparison of (a) runtime, (b) F-measure, 
%                  and (c) $DR_{Mean}$ (privacy) of our protocol (using the early mapping-based approach)
%                  with baseline approaches
%                  on \textbf{NCVR} corrupted (Corr) and non-corrupted (No-corr)
%                  datasets.}
%    }
%\label{fig:comp}
%\end{figure*}

\medskip
\noindent
\textbf{i. Comparison of different mapping:}
In Figure~\ref{fig:map}\,(a) we compare the runtime for 
our approach based on greedy (baseline), early and late mappings %one-to-one
(labelled as \textbf{GMap}, \textbf{EMap} and \textbf{LMap}, respectively), 
while in Figure~\ref{fig:map}\,(b) we compare their F-measure results
on the \textbf{NCVR} datasets.
The proposed early and late mappings require
similar or lower runtime than the baseline greedy mapping, and as expected the F-measure achieved with early and late mappings are significantly higher than greedy mapping.
Early mapping requires comparatively lower runtime than
late mapping at the cost of a small loss in linkage quality. %(lower F-measure results).
Since the loss in F-measure is not very significant, we use the early mapping-based approach as a default mapping in the rest of our experiments.

\begin{figure*}[ht!]
\centering
 \includegraphics[width=0.42\textwidth]{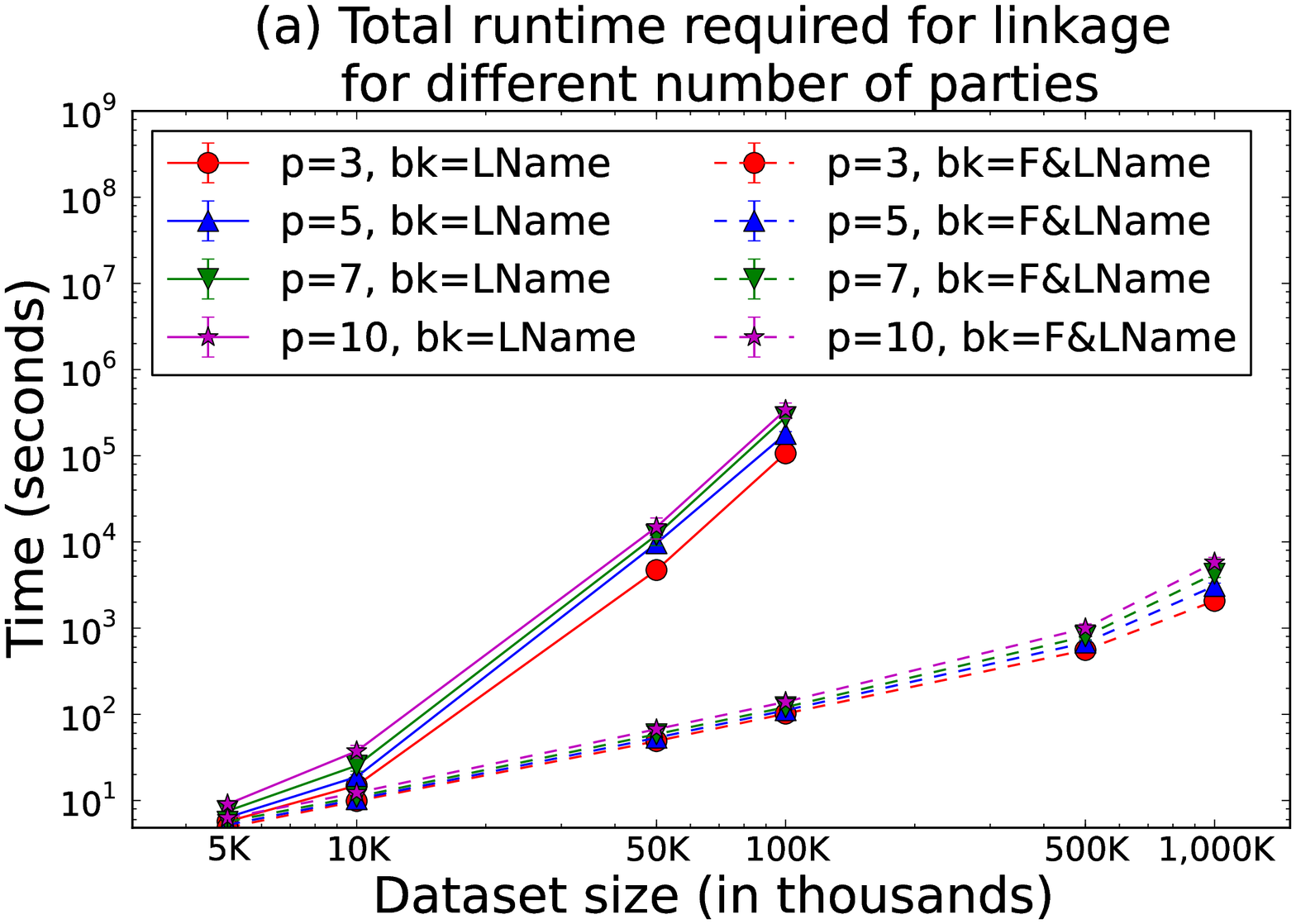}
~~~~~~~~~~
\includegraphics[width=0.42\textwidth]{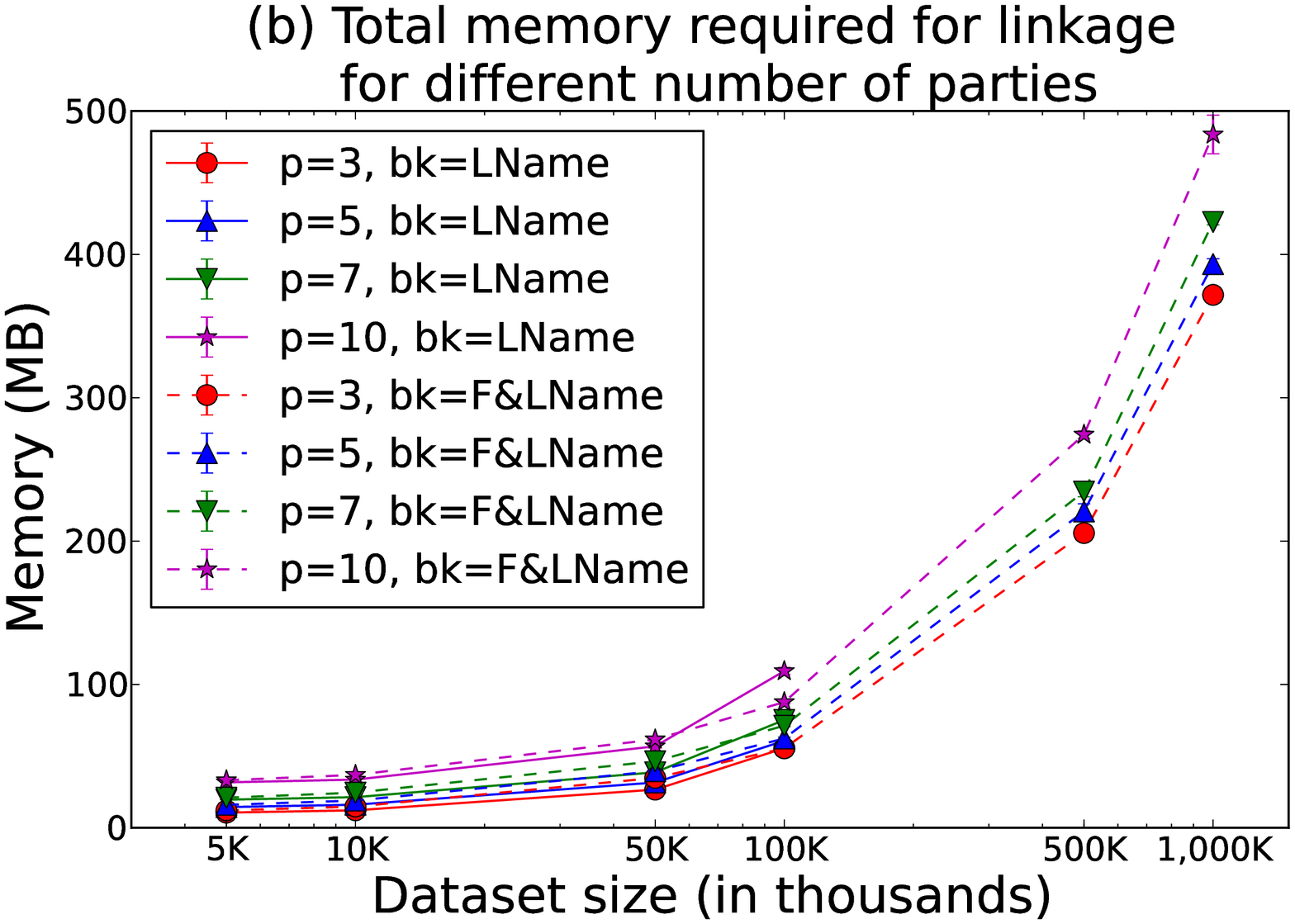}
 
  \caption{\small{Scalability results on different sizes of \textbf{NCVR} datasets in terms of (a) runtime and (b) memory size required for the linkage.}
    }
\label{fig:scal}
\end{figure*}

\medskip
\noindent
\textbf{ii. Similarity threshold vs. linkage quality:}
The F-measure achieved with different similarity thresholds on the \textbf{NCVR} datasets for $p=10$ is shown in Figure~\ref{fig:param}\,(a). The F-measure increases with the threshold up to $0.8$ on all datasets and then drops due to the loss in recall. As can be seen, the F-measure increases with larger thresholds on non-corrupted datasets (which require only exact matching) while it starts to decrease at a certain point on the corrupted datasets (which require approximate matching due to errors and variations). We therefore set the default threshold value to $0.8$ in our experiments. When the datasets are corrupted, the linkage quality becomes very low with increasing dataset sizes. These results indicate that more advanced classification techniques instead of a simple threshold-based classification are required to improve the linkage quality in the presence of real-world data errors~\cite{Chr12}.

\begin{figure*}[ht!]
\centering
 \includegraphics[width=0.42\textwidth]{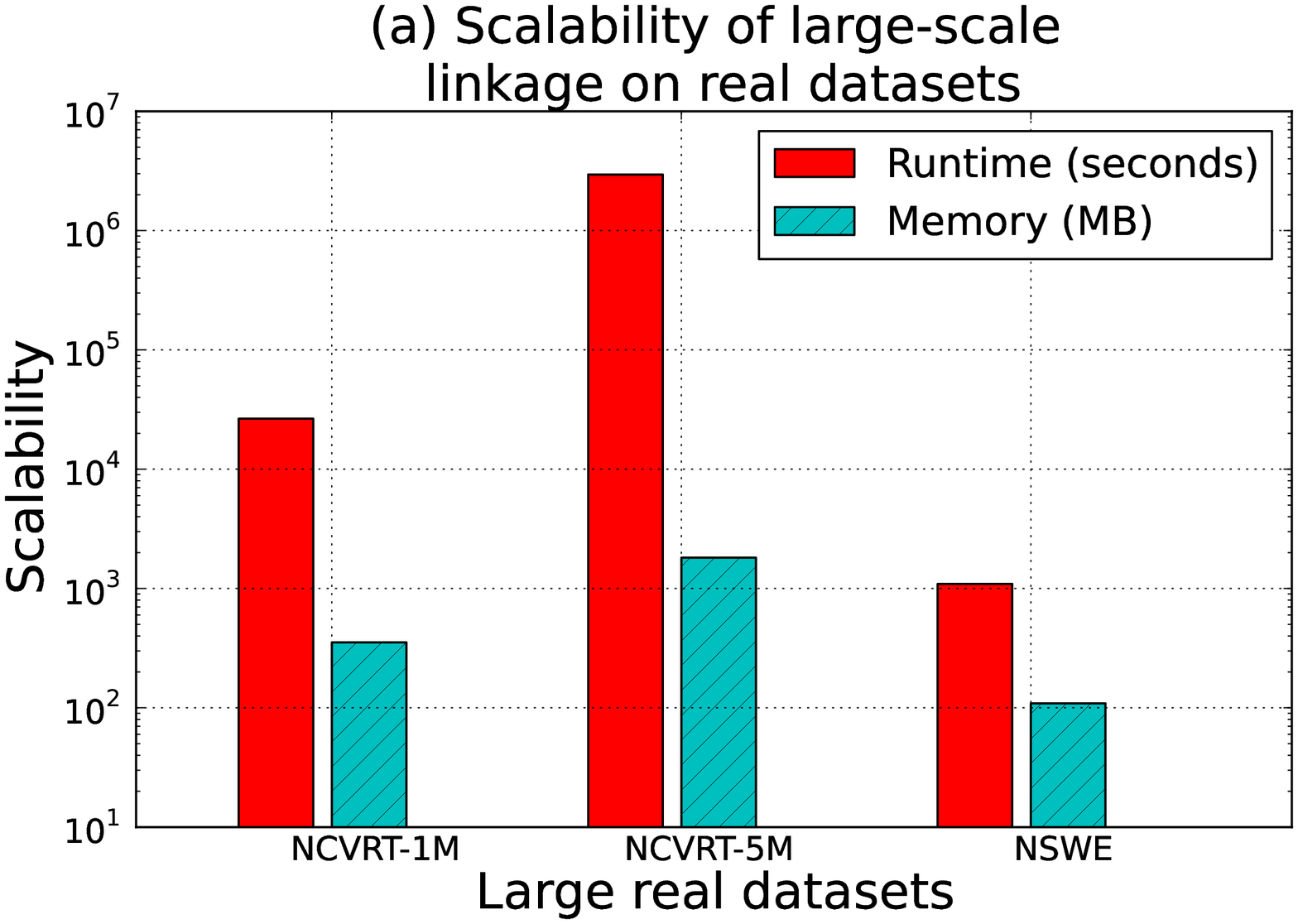}
~~~~~~~~~~
 \includegraphics[width=0.42\textwidth]{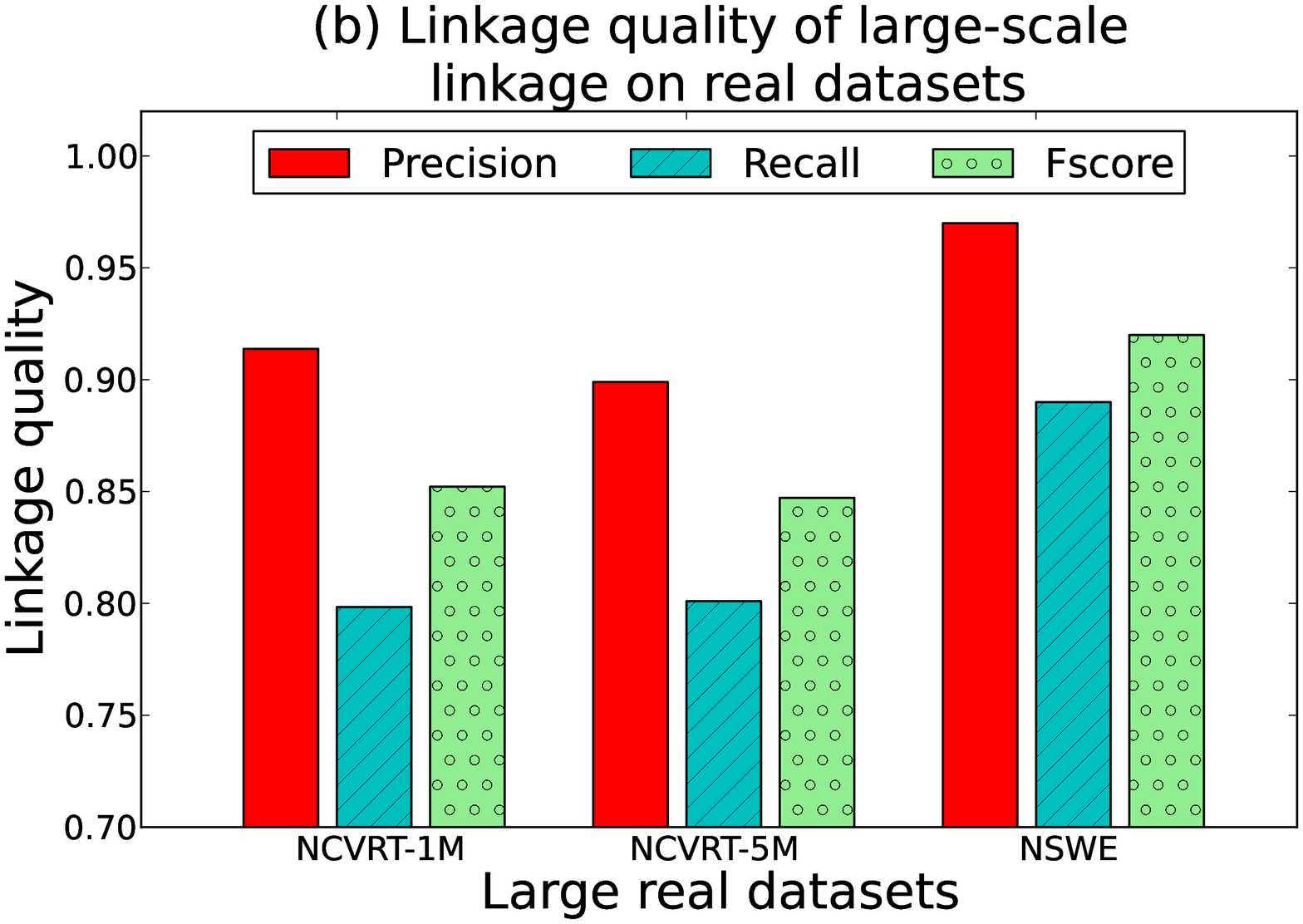}

  \caption{\small{Linkage results on real large-scale datasets (\textbf{NCVRT}-1M ($p=26$), \textbf{NCVRT}-5M ($p=26$) and \textbf{NSWE} ($p=5$)) in terms of (a) scalability and (b) linkage quality.}
    }
\label{fig:large-scale}
\end{figure*}

\medskip
\noindent
\textbf{iii. Minimum subset size vs. linkage quality:}
The F-measure of linkage achieved with different minimum subset sizes on the \textbf{NCVR} datasets for the early, late, and greedy mappings are shown
in Figure~\ref{fig:param}\,(b). Linking with smaller minimum subset size is more challenging than larger minimum subset size, as identifying matches across $2$ or more datasets is more difficult than all $10$ datasets, for example, due to the large number of combinations of datasets to be checked for matches. 
Our proposed mapping approaches outperform the greedy mapping significantly for smaller minimum subset sizes.

As expected, the linkage becomes more challenging with corrupted data using any mapping methods, 
when identifying records that match across
larger number of databases compared to smaller subsets of databases. With increasing number of databases, more corrupted records are included in the matches, resulting in significant loss of linkage quality. 
While data errors in real data are possible, the degree of corruption would be probably relatively low. According to the real NSWE and NCVRT datasets, the quality of data is very high with less than 1\% linkage errors when a probabilistic two-database matching technique is applied on the unencoded NSWE dataset~\cite{Ran18} and around 10\% error in the NCVRT dataset. We have tested relatively pessimistic scenarios by synthetically including 20\% and 40\% corruption to the matching records in NCVR datasets. The results on the corrupted datasets indicate that achieving high linkage quality in the presence of large amount of data errors is a big challenge, which needs to be mitigated through appropriate pre-processing techniques as well as clerical review possibly using active learning techniques~\cite{Vat17b}.
% (which is another challenge in PPRL)~\cite{Vat17b}. %Compared to the greedy mapping, our proposed early and late mappings perform significantly better for subset matching.

\medskip
\noindent
\textbf{iv. Scalability:}
We next evaluate the scalability of our protocol for different dataset sizes on the \textbf{NCVR} datasets in Figure~\ref{fig:scal}. 
%Figure~\ref{fig:scal}(b) shows the scalability of our approach on the \textbf{NCVR}
%datasets. 
%The total runtime required for the linkage as shown in Figure~\ref{fig:scal}(a)
%has a linear trend with the size of datasets and is scalable to linking multiple large databases. 
%Only around $5,000$ seconds are required to
%link $p=10$ databases each containing one million records, which is practical
%and feasible in real applications.
%Our approach does not require much memory for the 
%linkage as well, 
%%with a linear increase for linking larger datasets, 
%as shown in Figure~\ref{fig:scal}(b).
%%In Figure~\ref{fig:scal}(c), the runtime required by %the $LU$
%%(i.e.\ the incremental clustering step) also shows the %scalability
%%of linking multiple large databases.
When a combination of first and last name (labelled as \emph{FName} and \emph{LName}, respectively in the figure) attributes are used as blocking keys (\emph{BK}), the resulting block sizes become small, making our protocol highly scalable in terms of runtime and memory size to large datasets from multiple parties. However, when only the last name attribute is used as the BK our protocol shows a quadratic trend with the size and number of the datasets. The experiments on larger datasets of 500K and 1M with one attribute as BK required very large runtime due to the larger block sizes, and therefore we did not conduct this set of experiments due to time limitation. Advanced blocking and filtering techniques are therefore required to further reduce the computational complexity of large-scale multi-party linkage.

\begin{figure*}[th!]
\centering
 \includegraphics[width=0.42\textwidth]{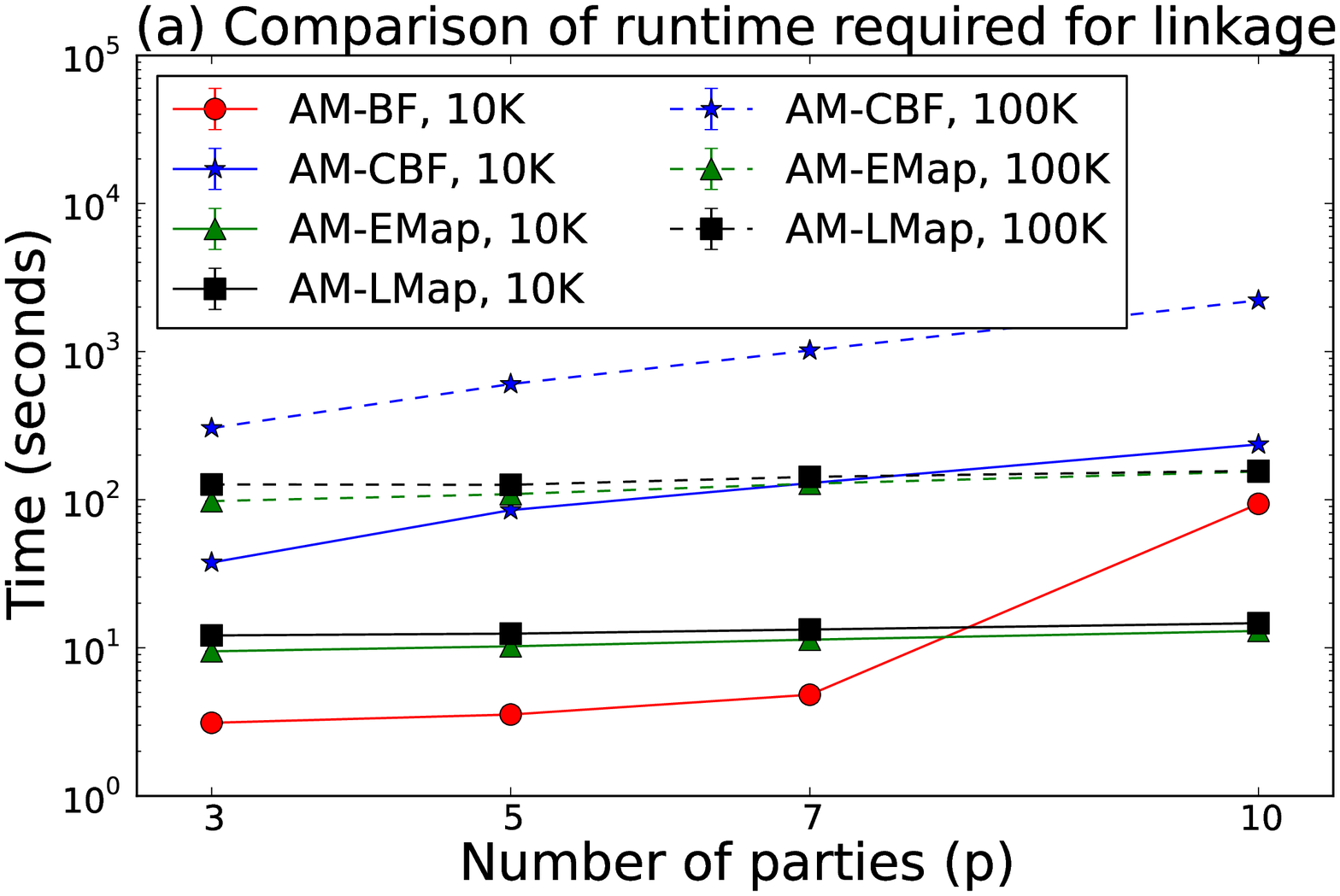}
~~~~~~~~~~
 \includegraphics[width=0.42\textwidth]{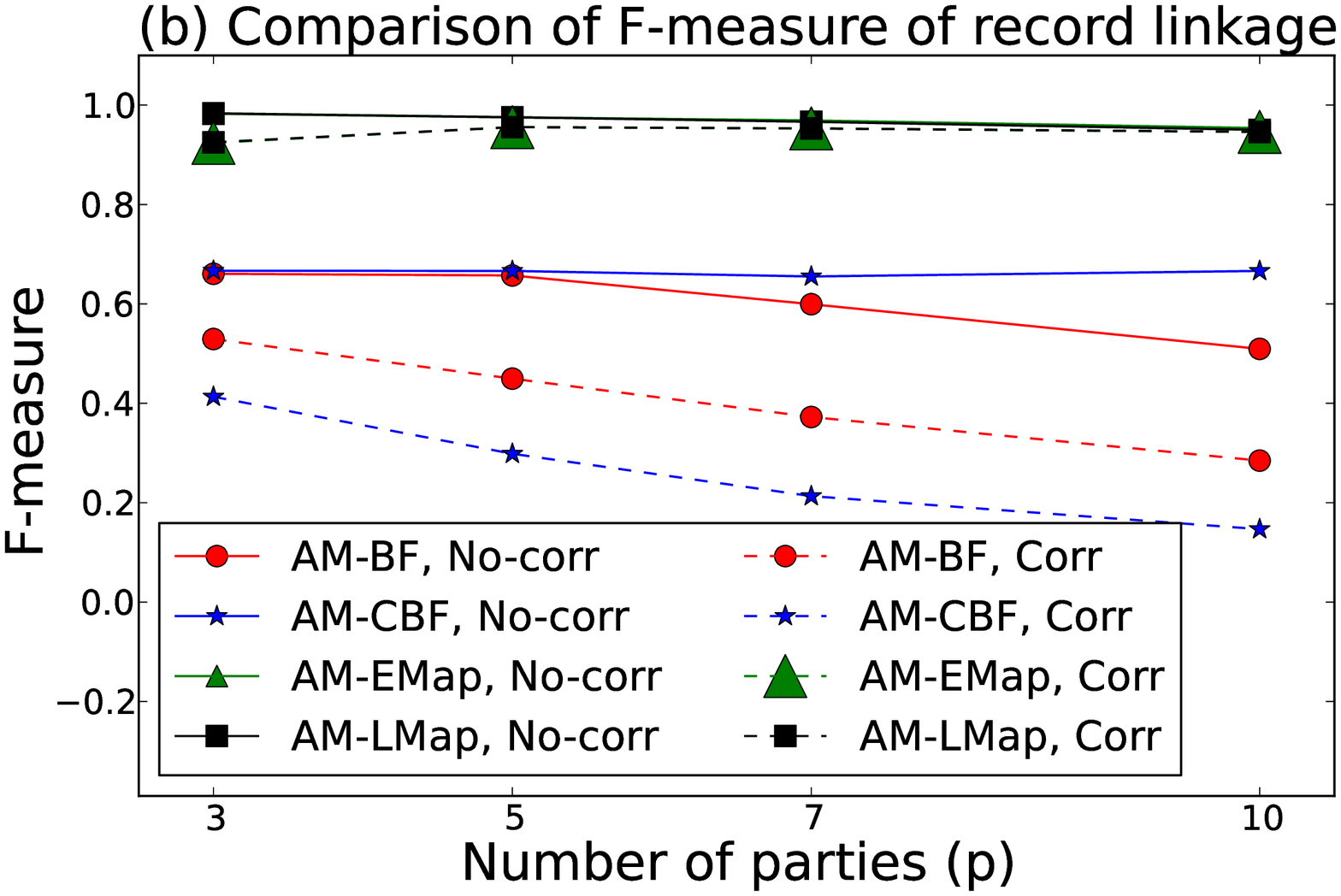}

  \caption{\small{Comparison of (a) runtime and (b) F-measure of our methods with baseline methods on \textbf{NCVR} corrupted (Corr) and non-corrupted (No-corr) datasets.}
    }
\label{fig:comp_baseline}
\end{figure*}

\medskip
\noindent
\textbf{v. Large-scale linkage:}
We conducted large-scale experiments of our approach
on the \textbf{NCVRT} and \textbf{NSWE} datasets. As can be seen
in Figures~\ref{fig:large-scale}\,(a) and (b), we are able to 
link multiple large datasets and achieve high linkage quality, which shows the viability of our approach for large-scale MP-PPRL applications. 
Since multi-database linkage requires an additional step of clustering (or mapping) after pair-wise matching, investigating other better clustering techniques that can achieve improved linkage quality is subject to further research. However, as shown in Figure~\ref{fig:large-scale}\,(a), the runtime required for linking such large multiple datasets is higher (even though it is significantly better compared to the baseline methods, as will be discussed below) and therefore more advanced computational methods, such as distributed computing and parallel processing, need to be investigated to further improve the efficiency of MP-PPRL.

\medskip
\noindent
\textbf{vi. Comparison with baseline:}
We next compare our approach with the %(labelled as \textbf{AM-Clus})
baseline approaches %(as described in Section~\ref{sec-related}) 
in Figure~\ref{fig:comp_baseline}
in terms of scalability and linkage quality.
As can be seen in Figure~\ref{fig:comp_baseline}\,(a), our approaches (\textbf{EMap} and \textbf{LMap}) 
require lower runtime for linking a large number of databases, 
where the runtime does not increase significantly with $p$
compared to \textbf{AM-BF}. %These two BF-based % \textbf{EM-BF} and 
%approaches 
The \textbf{AM-BF} approach requires lower 
runtime for linking smaller number of 
databases, however it increases exponentially with
larger $p$. We were unable to conduct experiments for this
approach on the \textbf{NCVR}-100K
datasets due to excessive memory consumption with the exponential
number of comparisons required by this approach.
The \textbf{AM-CBF} approach is more scalable than \textbf{AM-BF} %these two approaches
for linking a larger number of databases. %as shown in the figure.
This is because the improved communication patterns with CBF
reduce the exponential growth with $p$ down to the ring size $r$, where $r < p$~\cite{Vat16b}.
However, our proposed methods %\textbf{AM-Clus} %approach with both mappings 
require even lower runtime than \textbf{AM-CBF} and is more scalable with
increasing $p$.

\begin{figure*}[th!]
\centering
 \includegraphics[width=0.42\textwidth]{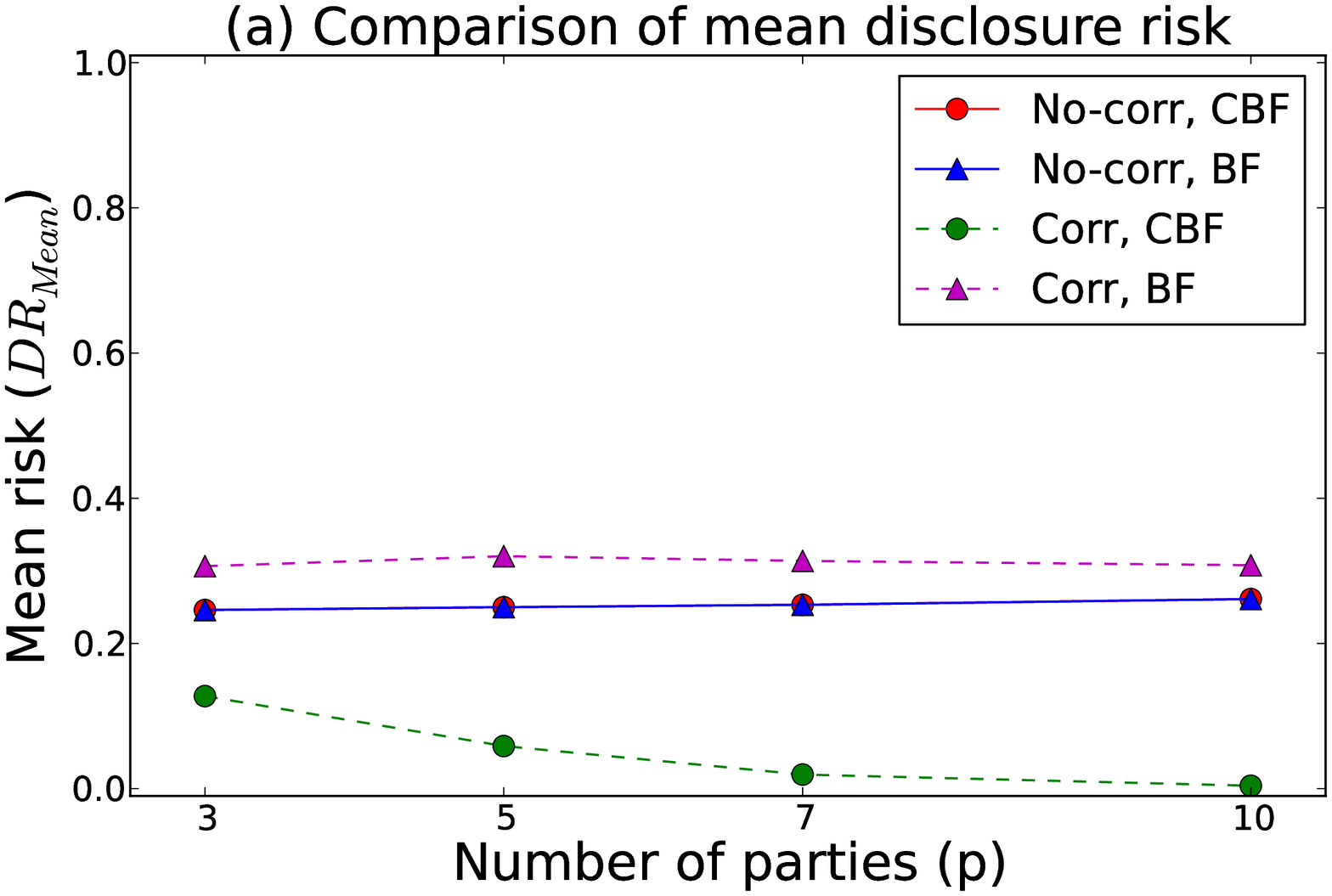}
~~~~~~~~~~
 \includegraphics[width=0.42\textwidth]{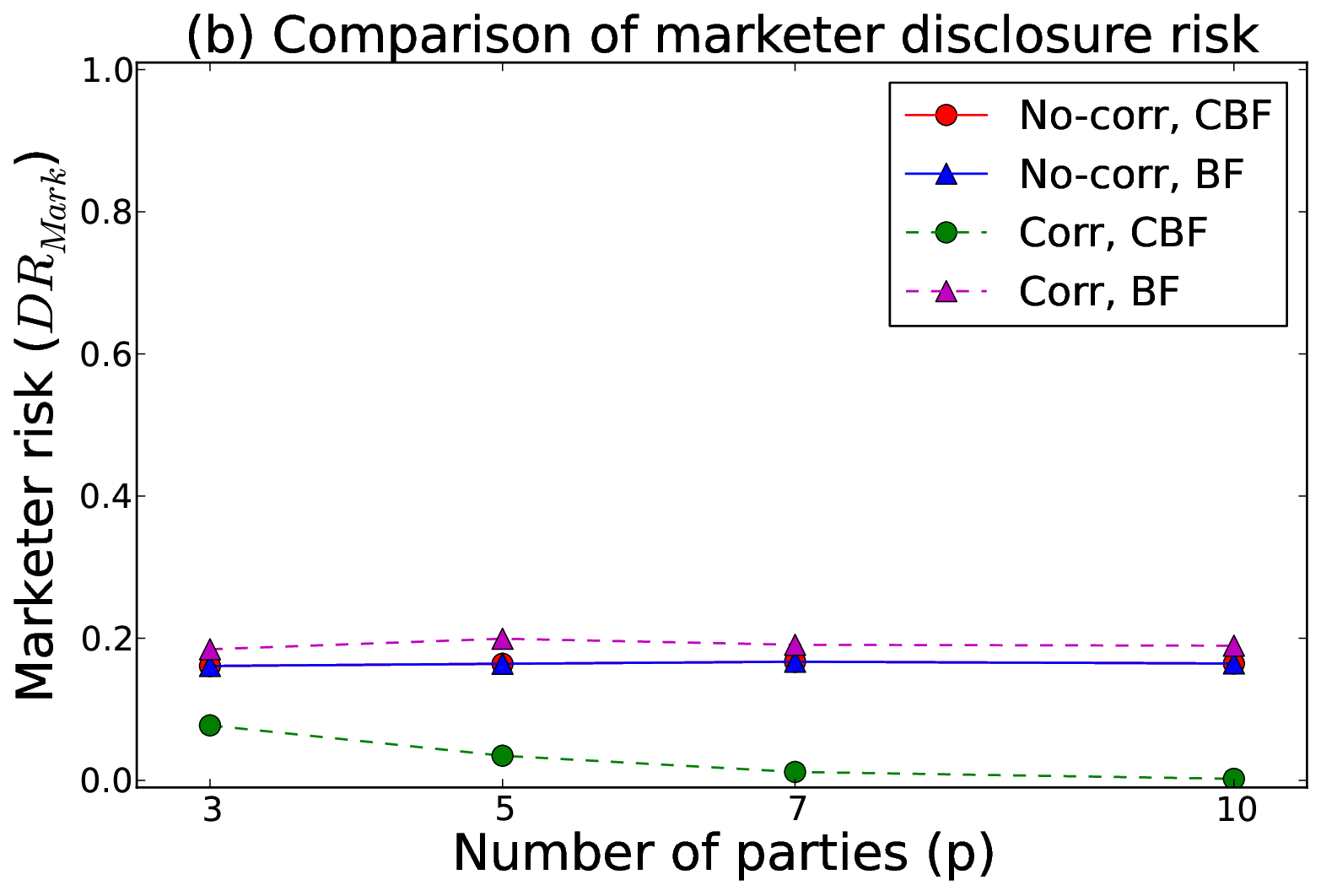}
  \caption{\small{Comparison of (a) mean and (b) marketer disclosure risk measures for BF and CBF encoding methods on \textbf{NCVR} corrupted (Corr) and non-corrupted (No-corr) datasets.}
    }
\label{fig:privacy}
\end{figure*}

%Our approach outperforms the baseline approaches in terms of linkage quality 
%by achieving high F-measure as shown in Figure~\ref{fig:comp}(b).
%The comparison of F-measure is shown in Figure~\ref{fig:comp}(b)
%We also compared the F-measure of all these approximate matching
%approaches with an exact matching MP-PPRL protocol~\cite{Lai06}, and
%as expected the approximate matching approaches outperform it
%on corrupted datasets.
%As expected, the \textbf{EM-BF} approach performs well on non-corrupted
%datasets, however on corrupted datasets all the approximate matching
%approaches outperform it. 
As shown in Figure~\ref{fig:comp_baseline}\,(b), our approaches (\textbf{EMap} and \textbf{LMap})
%the \textbf{AM-Clus}
%It 
achieve substantially higher F-measure results compared to all baseline methods
on both non-corrupted and corrupted datasets by
identifying matching records not only across all databases but also across
subsets of databases.
We also compared the F-measure of all these approximate matching
approaches with an exact matching MP-PPRL protocol~\cite{Lai06}, and
as expected the approximate matching approaches outperform it
on corrupted datasets.

\medskip
\noindent
\textbf{vii. Disclosure risk results:}
As shown in Figure~\ref{fig:privacy}, the CBF-based masking consistently has lower mean and marketer disclosure risk than BF-based masking (as we discussed in Section~\ref{subsec_privacy_analsis}). 
Therefore, CBF-based masking provides improved privacy than BF-based masking. This means that the same privacy results can be achieved by our protocol in terms of mean and marketer disclosure risks as with other CBF-based approaches~\cite{Vat16b} in the worst case.

This comparative evaluation shows that our \textbf{AM-Clus} approach outperforms existing approaches in terms of scalability and linkage quality,
while providing better/similar privacy results.

% --------------------------------------------------------------------

\section{Related Work}
\label{sec-related}

Various techniques have been proposed in the literature tackling the
problem of PPRL, as surveyed in~\cite{Vat17b,Vat13,Sch15,Tre08}. %Vat17b,
However, most of these approaches are limited to linking only two databases, 
and only few approaches have considered linking data
from multiple databases (MP-PPRL). 
%Existing MP-PPRL approaches 
%either use a $LU$ to perform the linkage on the masked databases
%sent to it by the database owners~\cite{Vat16b,Kan08,Kara15,Kee04}, 
%or they conduct the linkage among multiple
%databases without using a $LU$~\cite{Ran14,Ran16,Vat14c,Lai06}.
%
Neither of these techniques allow subset matching for MP-PPRL where records that match across subsets of databases are also identified in addition to records that match across all databases.

%Privacy techniques ranging from expensive 
%cryptographic-based Secure Multi-party Computation (SMC) techniques~\cite{Hal10}, such as
%homomorphic operations, 
%to efficient perturbation-based techniques, such as
%differential privacy, $k$-anonymity,
%and Bloom filters, have been used as masking functions for PPRL~\cite{Vat13}.

%\subsection{Exact matching / categorical data matching}
%\label{subsec:exact_matching}

A secure multi-party computation approach using an oblivious
transfer protocol was proposed
by O'Keefe et al.~\cite{Kee04} 
for PPRL on multiple
databases. While provably secure, the approach can only perform \emph{exact matching}
%of masked values 
(i.e.\ variations and errors in the QIDs
are not considered). %and it is \emph{computationally
%expensive} compared to efficient perturbation-based privacy techniques~\cite{Vat13}. 
%A multi-party $k$-anonymity-based approach
%was introduced in~\cite{Kan08}
Kantarcioglu et al.~\cite{Kan08} 
introduced a
MP-PPRL approach for \emph{categorical data} %that uses a $LU$ 
to perform secure equi-joins (\emph{exact matching}) 
%by the linkage unit 
on $k$-anonymous  
databases, where the QIDs of a record
are similar to at least $k$ other records
in the database~\cite{Swe02}. %, which is also an exact matching protocol.
%
%A MP-PPRL approach for \emph{exact matching} 
%of \emph{categorical values} %based on $k$-anonymity 
%was proposed by Mohammed et al.~\cite{Moh11}, where a top-down generalization is
%performed to provide $k$-anonymous privacy
%and the generalized blocks are then classified into matches
%and non-matches using a
%C4.5 decision tree classifier.
An \emph{exact matching} approach
for \emph{categorical data} was recently proposed
by Karapiperis et al.~\cite{Kara15}
using a Count-Min 
sketch data structure.
Sketches are used to summarize
the local set of elements which are then intersected to provide a 
global synopsis using homomorphic operations and symmetric noise
addition techniques~\cite{Clif02,Lin09}.
%As we will describe in detail in Section~\ref{sec-experiment}, 
%Lai et al.~\cite{Lai06} developed an
%efficient BF-based approach for MP-PPRL. However, the approach 
%performs only \emph{exact matching}.
%This approach has been adapted by Vatsalan and 
%Christen~\cite{Vat14c} for \emph{approximate
%matching} in MP-PPRL. %(as will be
%described in detail in Section~\ref{sec-experiment}).
Another \emph{exact matching} approach for MP-PPRL 
using Bloom filter (BF) encoding was introduced by Lai et al.~\cite{Lai06}, where a conjuncted BF is jointly constructed by all parties to identify matching records.
%A conjuncted BF is jointly constructed by
%all parties, such that each party processes a segment of the
%whole BF, which is then checked by parties individually
%against each of its records' BFs in order to classify the records 
%as matches or non-matches.

%\begin{figure*}[!th]
%  \centering
%  \scalebox{1.0}[1.0]{\includegraphics[width=0.9\textwidth]
%                      {PPRL_MP}}
%  \caption{\small{A general pipeline for PPRL on multiple ($p$) databases, as described in Section~\ref{sec-preli}.
%Databases are partitioned in step 1 to reduce the number of candidate record sets such that only records in the same blocks or partitions are compared and classified in step 2. Everything represented with a dotted outline needs to be privacy-preserved. }}
%           \label{fig:PPRL_MP}
%\end{figure*}

%\subsection{Approximate matching / string data matching}
%\label{subsec:approximate_matching}

All the MP-PPRL techniques described above are not practical %for record linkage 
in real applications as they allow only exact matching or matching of categorical data.
Vatsalan and Christen extended Lai et al.'s exact matching approach~\cite{Lai06}
to develop an \emph{approximate matching} solution for MP-PPRL~\cite{Vat14c}
by using BFs and a secure summation protocol~\cite{Clif02,Lin09} to
distributively calculate the similarity of a set of
BFs from different parties. 
A recent approach for \emph{approximate matching} in MP-PPRL based on Counting
Bloom filter (CBF) was proposed by Vatsalan et al.~\cite{Vat16b}. 
BFs from $p$ different databases were summarized into
a single CBF by applying a secure summation
protocol. %The advantage of using CBFs over BFs
%is that the information gain from a single CBF is
%lower than that from $p$ individual BFs,
%which improves privacy with no loss in accuracy~\cite{Vat16b}. 
Neither of these MP-PPRL approaches, however, supports identifying matching records in subset of parties.

Only limited grouping techniques have been developed to identify a set of matching records from multiple databases in the literature. Merge-based grouping simply groups or merges into one set all the records that have a similarity above the threshold~\cite{Ran14b}. The greedy best link approach proposed by Kendrick~\cite{Ken98} links each incoming record to the group that has the highest similarity with the incoming record.
An improved version of the best link approach was later proposed by Randall et al.~\cite{Ran15b}, which is referred to as weighted best link. In this approach all the records in the incoming file are first linked with the matching group of records, and then they are amalgamated according to the order of their weights. The advantage of the weighted best link approach is that it does not depend on the order of incoming records. However, the results depend on how the weights are calculated. Our proposed incremental clustering approaches are not only independent of the ordering of records %(however, they depend on the ordering of databases), 
but also the weights of links. 

%\subsection{Blocking / indexing}
%\label{subsec:blocking}

Scalability of PPRL has been addressed 
through the development of private blocking functions~\cite{Chr11,Ran15,Ran16}, and the more recently proposed summarization algorithms~\cite{Kara19}. 
%Ranbaduge et al.~\cite{Ran14} proposed 
%a family of efficient \emph{private blocking} techniques for MP-PPRL
%using BFs and bit-trees.  
However, the number of comparisons required for multi-party linkage
remains very large even with such existing private blocking and
filtering approaches employed~\cite{Vat14c,Ran14}. 
Recent work by Vatsalan et al.~\cite{Vat16b} proposed improved communication patterns
for reducing the number of comparisons for CBF-based
MP-PPRL. %In their approach, the parties are grouped into rings
%and a secure summation protocol is employed to generate a CBF
%for each set of parties' BFs. The matching of records
%%matching 
%is conducted either (1) sequentially, such that only the matches 
%(whose similarity
%calculated using their CBF is above a threshold) 
%of a ring are compared with the candidate record sets of the next ring, or 
%(2) symmetrically, where matches are identified 
%for each individual ring %in the
%first, % phase, 
%and then using the matches from individual rings the
%matches across all rings are identified. % in the second phase.
%
The na\"{\i}ve computation complexity of MP-PPRL techniques is exponential
in the number of records per database ($n^p$, assuming $n$ 
records in each of the $p$ databases). The improved communication patterns developed
by Vatsalan et al.~\cite{Vat16b} reduce this exponential growth with $p$ down 
to the ring size $r$ (with $r < p$).
In contrast, our proposed approach efficiently performs subset matching with a
quadratic computation
complexity in the size and number of databases ($O(n^2 \cdot p^2)$), which allows large-scale
MP-PPRL.

% --------------------------------------------------------------------

\section{Conclusion}
\label{sec-conclusion}
We have presented a scalable MP-PPRL protocol
that is highly efficient for practical applications, such as health data linkage, and it improves the linkage quality
compared to existing MP-PPRL approaches that 
only allow identifying matching records across all
databases and do not support subset matching.
Our protocol uses graph-based incremental clustering
to achieve efficient identification of matching records
across all and subsets of large databases.

An experimental evaluation conducted on large real
datasets (including 26 voter registration databases
each containing over 5 million records and 5 real
emergency admissions datasets each containing
around 700,000 records) shows that our approach is practical
for real large-scale MP-PPRL applications.
Our approach outperforms existing MP-PPRL approaches in terms of linkage quality and scalability. %while
%providing similar privacy protection.

In future work, we aim to investigate how
existing clustering algorithms for record
linkage~\cite{Has09c,Nan19,Sae18,Cha11,Yi17} can be adapted for MP-PPRL. 
One important direction of this work is to study incremental clustering for dynamic data matching in MP-PPRL~\cite{Li10}.
We also plan to evaluate the impact of pre-processing techniques, especially dealing with missing values~\cite{Ani19}, on the performance of privacy-preserving clustering.
Another direction is to study how incremental
clustering can be parallelized to improve scalability
of large-scale MP-PPRL. %We plan to conduct more
%experimental studies with different blocking, ordering, %mapping, 
%and similarity
%functions to see how they affect linkage quality and efficiency.
Investigating other advanced mapping, encoding, similarity, and classification functions (including relational clustering and collective classification~\cite{Chr12}) for clustering-based MP-PPRL would also be interesting directions for future work.
%Finally, we plan to investigate improved
%classification techniques %for MP-PPRL
%including relational clustering and collective classification, %~\cite{Chr12},
%which are successfully used in
%non-PPRL applications.

% --------------------------------------------------------------------

\section{Acknowledgements}

This work was partially funded by the Australian Research Council
under Discovery Projects DP130101801 and DP160101934, and
Universities Australia and the German Academic Exchange Service
(DAAD). We would like to thank Sean Randall from the Centre for
Data Linkage, Curtin University, for conducting experiments
using our proposed methods on the sensitive NSWE dataset.





%\begin{thebibliography}{00}
%\bibitem{}
%
%\end{thebibliography}

%\balance
\section{References} 

\bibliographystyle{elsarticle-num}
\bibliography{paper} 

\begin{thebibliography}{10}
\expandafter\ifx\csname url\endcsname\relax
  \def\url#1{\texttt{#1}}\fi
\expandafter\ifx\csname urlprefix\endcsname\relax\def\urlprefix{URL }\fi
\expandafter\ifx\csname href\endcsname\relax
  \def\href#1#2{#2} \def\path#1{#1}\fi

\bibitem{Vat17b}
D.~Vatsalan, Z.~Sehili, P.~Christen, E.~Rahm, Privacy-preserving record linkage
  for {B}ig data: Current approaches and research challenges, in: Handbook of
  Big Data Technologies, Springer, 2017, pp. 851--895.

\bibitem{Clif02}
C.~Clifton, M.~Kantarcioglu, J.~Vaidya, X.~Lin, M.~Y. Zhu, Tools for privacy
  preserving distributed data mining, SIGKDD Explorations 4~(2) (2002) 28--34.

\bibitem{Che19}
S.~L. Cheah, V.~L. Scarf, C.~Rossiter, C.~Thornton, C.~S. Homer, Creating the
  first national linked dataset on perinatal and maternal outcomes in
  australia: Methods and challenges, Journal of Biomedical Informatics (2019)
  103152.

\bibitem{Chr12}
P.~Christen, Data matching - concepts and techniques for record linkage, entity
  resolution, and duplicate detection, Data-Centric Systems and Applications,
  Springer, 2012.

\bibitem{Chi17}
Y.~Chi, J.~Hong, A.~Jurek, W.~Liu, D.~O’Reilly, Privacy preserving record
  linkage in the presence of missing values, Information Systems 71 (2017)
  199--210.

\bibitem{Dur13}
E.~A. Durham, C.~Toth, M.~Kuzu, M.~Kantarcioglu, Y.~Xue, B.~Malin, Composite
  {B}loom filters for secure record linkage, IEEE Transactions on Knowledge and
  Data Engineering 26~(12) (2014) 2956--2968.

\bibitem{Kara15}
D.~Karapiperis, D.~Vatsalan, V.~S. Verykios, P.~Christen, Large-scale
  multi-party counting set intersection using a space efficient global
  synopsis, in: Database Systems for Advanced Applications, Hanoi, 2015.

\bibitem{Vat13}
D.~Vatsalan, P.~Christen, V.~S. Verykios, A taxonomy of privacy-preserving
  record linkage techniques, Information Systems 38~(6) (2013) 946--969.

\bibitem{Con04}
J.~R. Condon, T.~Barnes, J.~Cunningham, B.~K. Armstrong, Long-term trends in
  cancer mortality for indigenous australians in the northern territory,
  Medical Journal of Australia 180~(10) (2004) 504.

\bibitem{Kue11}
C.~E. Kuehni, C.~S. Rueegg, G.~Michel, C.~E. Rebholz, M.-P.~F. Strippoli, F.~K.
  Niggli, M.~Egger, N.~X. von~der Weid, S.~P. O.~G. (SPOG), Cohort profile: the
  {S}wiss childhood cancer survivor study, International journal of
  epidemiology 41~(6) (2011) 1553--1564.

\bibitem{Bak18}
D.~Baker, B.~M. Knoppers, M.~Phillips, D.~van Enckevort, P.~Kaufmann,
  H.~Lochmuller, D.~Taruscio, Privacy-preserving linkage of genomic and
  clinical data sets, IEEE/ACM Transactions on Computational Biology and
  Bioinformatics (2018) 1.

\bibitem{ONS13}
{Office for National Statistics}, Matching anonymous data, in: Beyond 2011,
  2013.

\bibitem{Phu12}
C.~Phua, K.~Smith-Miles, V.~C. Lee, R.~Gayler, Resilient identity crime
  detection, IEEE Transactions on Knowledge and Data Engineering 24~(3) (2012)
  533.

\bibitem{Kop10}
H.~K{\"o}pcke, E.~Rahm, Frameworks for entity matching: A comparison, Data \&
  Knowledge Engineering 69~(2) (2010) 197--210.

\bibitem{Wan07}
D.-W. Wang, C.-J. Liau, T.-s. Hsu, An epistemic framework for privacy
  protection in database linking, Data \& Knowledge Engineering 61~(1) (2007)
  176--205.

\bibitem{Has09c}
O.~Hassanzadeh, F.~Chiang, H.~C. Lee, R.~J. Miller, Framework for evaluating
  clustering algorithms in duplicate detection, Proceedings of the Very Large
  Database Endowment 2~(1) (2009) 1282--1293.

\bibitem{Nan19}
C.~Nanayakkara, P.~Christen, T.~Ranbaduge, Robust temporal graph clustering for
  group record linkage, in: Pacific-Asia Conference on Knowledge Discovery and
  Data Mining, Macau, 2019.

\bibitem{Sae18}
A.~Saeedi, M.~Nentwig, E.~Peukert, E.~Rahm, Scalable matching and clustering of
  entities with famer, Complex Systems Informatics and Modeling Quarterly~(16)
  (2018) 61--83.

\bibitem{Sch09}
R.~Schnell, T.~Bachteler, J.~Reiher, Privacy-preserving record linkage using
  {Bloom} filters, BMC Medical Informatics and Decision Making 9~(1) (2009) 1.

\bibitem{Chr18}
P.~Christen, A.~Vidanage, T.~Ranbaduge, R.~Schnell, Pattern-mining based
  cryptanalysis of {B}loom filters for privacy-preserving record linkage, in:
  PAKDD, Springer LNAI, Melbourne, 2018, pp. 530--542.

\bibitem{Chr18b}
P.~Christen, T.~Ranbaduge, D.~Vatsalan, R.~Schnell, Precise and fast
  cryptanalysis for {B}loom filter based privacy-preserving record linkage,
  IEEE Transactions on Knowledge and Data Engineering (2018) 1.

\bibitem{Chr11}
P.~Christen, A survey of indexing techniques for scalable record linkage and
  deduplication, IEEE Transactions on Knowledge and Data Engineering 24~(9)
  (2012) 1537--1555.

\bibitem{Vat14c}
D.~Vatsalan, P.~Christen, Scalable privacy-preserving record linkage for
  multiple databases, in: ACM Conference in Knowledge Management, Shanghai,
  2014.

\bibitem{Vat16b}
D.~Vatsalan, P.~Christen, E.~Rahm, Scalable privacy-preserving linking of
  multiple databases using counting {B}loom filters, in: Workshop on Privacy
  and Discrimination in Data Mining held at IEEE ICDM, Barcelona, 2016.

\bibitem{Rah16}
E.~Rahm, The case for holistic data integration, in: Advances in Databases and
  Information Systems, Springer, 2016, pp. 11--27.

\bibitem{Ran15b}
S.~M. Randall, J.~H. Boyd, A.~M. Ferrante, A.~P. Brown, J.~B. Semmens, Grouping
  methods for ongoing record linkage, in: KDD Workshop on Population
  Informatics, Sydney, 2015.

\bibitem{Ken98}
S.~Kendrick, M.~Douglas, D.~Gardner, D.~Hucker, Best-link matching of
  {S}cottish health data sets., Methods of Information in Medicine 37~(1)
  (1998) 64--68.

\bibitem{Lai06}
P.~Lai, S.-M. Yiu, K.~Chow, C.~Chong, L.~Hui, {An Efficient {B}loom filter
  based Solution for Multiparty Private Matching}, in: Security and Management,
  Las Vegas, 2006.

\bibitem{Nau10}
F.~Naumann, M.~Herschel, An introduction to duplicate detection, Synthesis
  Lectures on Data Management 2~(1).

\bibitem{Ran14}
T.~Ranbaduge, P.~Christen, D.~Vatsalan, Tree based scalable indexing for
  multi-party privacy-preserving record linkage, in: Australasian Data Mining,
  Brisbane, 2014.

\bibitem{All05}
A.~Al-Lawati, D.~Lee, P.~McDaniel, Blocking-aware private record linkage, in:
  IQIS, 2005, pp. 59--68.

\bibitem{Seh15}
Z.~Sehili, L.~Kolb, C.~Borgs, R.~Schnell, E.~Rahm, Privacy preserving record
  linkage with {PPJoin}, in: BTW Conference, Hamburg, 2015.

\bibitem{Ran13}
S.~M. Randall, A.~M. Ferrante, J.~H. Boyd, J.~B. Semmens, Privacy-preserving
  record linkage on large real world datasets, Journal of Biomedical
  Informatics 50~(1) (2014) 1.

\bibitem{Bro19}
A.~P. Brown, S.~M. Randall, J.~H. Boyd, A.~M. Ferrante, Evaluation of
  approximate comparison methods on bloom filters for probabilistic linkage,
  International Journal of Population Data Science 4~(1).

\bibitem{Sch15}
R.~Schnell, Privacy preserving record linkage, in: K.~Harron, H.~Goldstein,
  C.~Dibben (Eds.), Methodological developments in data linkage, Wiley,
  Chichester, 2016, pp. 201--225.

\bibitem{Vat16}
D.~Vatsalan, P.~Christen, Privacy-preserving matching of similar patients,
  Journal of Biomedical Informatics 59 (2016) 285--298.

\bibitem{Kara17}
D.~Karapiperis, A.~Gkoulalas-Divanis, V.~S. Verykios, Distance-aware encoding
  of numerical values for privacy-preserving record linkage, in: International
  Conference on Data Engineering, San Diego, 2017, pp. 135--138.

\bibitem{Vat12}
D.~Vatsalan, P.~Christen, An iterative two-party protocol for scalable
  privacy-preserving record linkage, in: Australasian Data Mining Conference,
  Sydney, 2012.

\bibitem{Ack14}
M.~Ackerman, S.~Dasgupta, Incremental clustering: The case for extra clusters,
  in: Advances in Neural Information Processing Systems, 2014, pp. 307--315.

\bibitem{Kuh55}
H.~W. Kuhn, The hungarian method for the assignment problem, Naval research
  logistics quarterly 2~(1-2) (1955) 83--97.

\bibitem{Nen18}
M.~Nentwig, E.~Rahm, Incremental clustering on linked data, in: Workshop on
  Data Integration and Application held at IEEE ICDM, Singapore, 2018.

\bibitem{Kuz11}
M.~Kuzu, M.~Kantarcioglu, E.~Durham, B.~Malin, A constraint satisfaction
  cryptanalysis of {Bloom} filters in private record linkage, in: Privacy
  Enhancing Technologies Symposium, Waterloo, Canada, 2011, pp. 226--245.

\bibitem{Nie14}
F.~Niedermeyer, S.~Steinmetzer, M.~Kroll, R.~Schnell, Cryptanalysis of basic
  {Bloom} filters used for privacy preserving record linkage, Journal of
  Privacy and Confidentiality 6~(2) (2014) 59--79.

\bibitem{Lin09}
Y.~Lindell, B.~Pinkas, Secure multiparty computation for privacy-preserving
  data mining, Journal of Privacy and Confidentiality 1~(1) (2009) 1.

\bibitem{Vat14}
D.~Vatsalan, P.~Christen, C.~M. O'Keefe, V.~S. Verykios, An evaluation
  framework for privacy-preserving record linkage, Journal of Privacy and
  Confidentiality 6~(1) (2014) 1.

\bibitem{Kuz13}
M.~Kuzu, M.~Kantarcioglu, A.~Inan, E.~Bertino, E.~Durham, B.~Malin, Efficient
  privacy-aware record integration, in: ACM International Conference on
  Extending Database Technology, Genoa, Italy, 2013, pp. 167--178.

\bibitem{Tas13}
T.~Tassa, D.~J. Cohen, Anonymization of centralized and distributed social
  networks by sequential clustering, IEEE Transactions on Knowledge and Data
  Engineering 25~(2) (2013) 311--324.

\bibitem{Tra13}
K.-N. Tran, D.~Vatsalan, P.~Christen, {GeCo}: an online personal data generator
  and corruptor, in: ACM Conference in Knowledge Management, San Francisco,
  2013, pp. 2473--2476.

\bibitem{Chr13NC}
P.~Christen, Preparation of a real voter data set for record linkage and
  duplicate detection research, Tech. rep., Research School of Computer
  Science, Australian National University (2014).

\bibitem{Ran16b}
T.~Ranbaduge, D.~Vatsalan, S.~Randall, P.~Christen, Evaluation of advanced
  techniques for multi-party privacy-preserving record linkage on real-world
  health databases, in: International Population Data Linkage Conference,
  Swansea, Wales, 2016.

\bibitem{Han18}
D.~Hand, P.~Christen, A note on using the {F}-measure for evaluating record
  linkage algorithms, Statistics and Computing 28~(3) (2018) 539--547.

\bibitem{Ran18}
S.~Randall, A.~Brown, J.~Boyd, R.~Schnell, C.~Borgs, A.~Ferrante,
  Sociodemographic differences in linkage error: an examination of four
  large-scale datasets, BMC Health Services Research 18~(1) (2018) 678.

\bibitem{Tre08}
S.~Trepetin, Privacy-preserving string comparisons in record linkage systems: a
  review, Information Security Journal: A Global Perspective 17~(5) (2008)
  253--266.

\bibitem{Kee04}
C.~M. O'Keefe, M.~Yung, L.~Gu, R.~Baxter, Privacy-preserving data linkage
  protocols, in: ACM Workshop on Privacy in the Electronic Society, Washington,
  2004.

\bibitem{Kan08}
M.~Kantarcioglu, W.~Jiang, B.~Malin, A privacy-preserving framework for
  integrating person-specific databases, in: Privacy in Statistical Databases,
  Istanbul, 2008, pp. 298--314.

\bibitem{Swe02}
L.~Sweeney, K-anonymity: A model for protecting privacy, International Journal
  of Uncertainty Fuzziness and Knowledge Based Systems 10~(5) (2002) 557--570.

\bibitem{Ran14b}
S.~M. Randall, J.~H. Boyd, A.~M. Ferrante, J.~K. Bauer, J.~B. Semmens, Use of
  graph theory measures to identify errors in record linkage, Computer Methods
  and Programs in Biomedicine 115~(2) (2014) 55--63.

\bibitem{Ran15}
T.~Ranbaduge, P.~Christen, D.~Vatsalan, Clustering-based scalable indexing for
  multi-party privacy-preserving record linkage, in: PAKDD, Springer LNAI,
  Hanoi, 2015.

\bibitem{Ran16}
T.~Ranbaduge, D.~Vatsalan, P.~Christen, Hashing-based distributed multi-party
  blocking for privacy-preserving record linkage, in: PAKDD, Springer LNAI,
  Auckland, 2016.

\bibitem{Kara19}
D.~Karapiperis, A.~Gkoulalas-Divanis, V.~S. Verykios, Summarizing and linking
  electronic health records, Distributed and Parallel Databases (2019) 1--40.

\bibitem{Cha11}
S.~Chakraborty, N.~Nagwani, Analysis and study of incremental k-means
  clustering algorithm, in: High Performance Architecture and Grid Computing,
  Springer, 2011, pp. 338--341.

\bibitem{Yi17}
H.~Yin, A.~R. Benson, J.~Leskovec, D.~F. Gleich, Local higher-order graph
  clustering, in: ACM SIGKDD International Conference on Knowledge Discovery
  and Data Mining, ACM, 2017, pp. 555--564.

\bibitem{Li10}
Z.~Li, J.-G. Lee, X.~Li, J.~Han, Incremental clustering for trajectories, in:
  International Conference on Database Systems for Advanced Applications,
  Springer, 2010, pp. 32--46.

\bibitem{Ani19}
I.~C. Anindya, M.~Kantarcioglu, B.~Malin, Determining the impact of missing
  values on blocking in record linkage, in: PAKDD, Springer LNAI, Macau, 2019,
  pp. 262--274.

\end{thebibliography}

\balance

\end{document}